\documentclass{article}
\usepackage[english]{babel}
\usepackage[a4paper]{geometry}
\usepackage{microtype,xstring}
\usepackage{amssymb,amsmath,amsthm}
\usepackage{graphics,graphicx,color}
\usepackage{enumerate,paralist,hyperref}
\usepackage{caption,subcaption,xspace}

\newtheorem{theorem}{Theorem}
\newtheorem{lemma}{Lemma}
\newtheorem{claim}{Claim}
\newtheorem{crl}{Corollary}
\newtheorem{property}{Property}

\newcommand{\pe}{potential edge\xspace}
\newcommand{\pes}{potential edges\xspace}
\newcommand{\pec}{potential empty cycle\xspace}

\newcommand{\side}{side\xspace}
\newcommand{\se}{side-edge\xspace}
\newcommand{\ses}{side-edges\xspace}
\newcommand{\sa}{side-apart\xspace}

\title{On Optimal 2- and 3-Planar Graphs}

\author{Michael~A.~Bekos$^1$, Michael~Kaufmann$^1$, Chrysanthi~N.~Raftopoulou$^2$
\\
\medskip
\\
$^1$Wilhelm-Schickhard-Institut f\"ur Informatik, Universit\"at T\"ubingen, Germany\\
\texttt{\{bekos,mk\}@informatik.uni-tuebingen.de}
\\
$^2$School of Applied Mathematical \& Physical Sciences, NTUA, Greece\\
\texttt{crisraft@mail.ntua.gr}}
\date{}

\begin{document}

\maketitle

\begin{abstract}
A graph is $k$-planar if it can be drawn in the plane such that no edge is crossed more than $k$ times. While for $k=1$, \emph{optimal} $1$-planar graphs, i.e.\ those with $n$ vertices and exactly $4n-8$ edges, have been completely characterized, this has not been the case for $k \geq 2$. For $k=2,3$ and $4$, upper bounds on the edge density have been developed for the case of simple graphs by Pach and T\'oth, Pach et al.\ and Ackerman, which have been used to improve the well-known ``Crossing Lemma''. Recently, we proved that these bounds also apply to non-simple $2$- and $3$-planar graphs without homotopic parallel edges and self-loops.

In this paper, we completely characterize optimal $2$- and $3$-planar graphs, i.e., those that achieve the aforementioned upper bounds. We prove that they have a remarkably simple regular structure, although they might be non-simple. The new characterization allows us to develop notable insights concerning new inclusion relationships with other graph classes.
\end{abstract}

% ============================================================================
\section{Introduction}
\label{sec:introduction}
% ============================================================================

Topological graphs, i.e.\ graphs that usually come with a representation of the edges as Jordan arcs between corresponding vertex points in the plane, form a well-established subject in the field of geometric graph theory. Besides the classical problems on crossing numbers and crossing configurations \cite{DBLP:journals/dcg/AlonE89,DBLP:journals/dcg/LovaszPS97,DBLP:journals/jgt/Turan77}, the well-known ''Crossing Lemma'' \cite{ACNS82,Lei83} stands out as a prominent result. Researchers on graph drawing have followed a slightly different research direction, based on extensions of planar graphs that allow crossings in some restricted local configurations~\cite{DBLP:journals/tcs/BinucciGDMPST15,DBLP:journals/algorithmica/CheongHKK15,DBLP:journals/tcs/DidimoEL11,DBLP:journals/cj/GiacomoDLMW15,KU14}. The main focus has been on \emph{1-planar graphs}, where each edge can be crossed at most once, with early results dating back to Ringel~\cite{Ringel65} and Bodendiek et al.~\cite{BSW84}. Extensive work on generation~\cite{DBLP:journals/siamdm/Suzuki10}, characterization~\cite{DBLP:conf/cocoon/HongELP12}, recognition~\cite{DBLP:journals/corr/Brandenburg16a}, coloring~\cite{Borodin95}, page number~\cite{BB0R15}, etc.\ has led to a very good understanding of structural properties of 1-planar graphs.

Pach and T\'{o}th~\cite{PachT97}, Pach et al.~\cite{PachRTT06}~and Ackerman~\cite{DBLP:journals/corr/Ackerman15} bridged the two research directions by considering the more general class of \emph{$k$-planar graphs}, where each edge is allowed to be crossed at most $k$ times. In particular, Pach and T\'{o}th provided significant progress, as they developed techniques for upper bounds on the number of edges of simple $k$-planar graphs, which subsequently led to upper bounds of $5n -10$~\cite{PachT97}, $5.5n - 11$~\cite{PachRTT06} and $6n-12$~\cite{DBLP:journals/corr/Ackerman15} for simple $2$-, $3$- and $4$-planar graphs, respectively. An interesting consequence was the improvement of the leading constant in the ''Crossing Lemma''. Note that for general $k$, the current best bound on the number of edges is $4.1 \sqrt k n$~\cite{PachT97}.

Recently, we generalized the result and the bound of Pach et al.~\cite{PachRTT06} to non-simple graphs, where non-homotopic parallel edges as well as non-homotopic self-loops are allowed~\cite{BKR16}. Note that this non-simplicity extension is quite natural and not new, as for planar graphs, the density bound of $3n-6$ still holds for such non-simple graphs.

In this paper, we now completely characterize optimal non-simple $2$- and $3$-planar graphs, i.e.\ those that achieve the bounds of $5n-10$ and $5.5n-11$ on the number of edges, respectively; refer to Theorems~\ref{thm:2-characterization} and \ref{thm:3-characterization}. In particular, we prove that the commonly known $2$-planar graphs achieving the upper bound of $5n-10$ edges, are in fact, the only optimal $2$-planar graphs. Such graphs consist of a crossing-free subgraph where all not necessarily simple faces have size $5$. At each face there are $5$ more edges crossing in its interior. We correspondingly show that the optimal $3$-planar graphs have a similar simple and regular structure where each planar face has size $6$ and contains $8$ additional crossing edges.

The remainder of this paper is structured as follows: In Section~\ref{sec:preliminaries} we introduce preliminary notions and notation. In Section~\ref{sec:properties} we present several structural properties of optimal $2$- and $3$-planar graphs that we use in Sections~\ref{sec:2planar} and \ref{sec:3planar} in order to give their characterizations. We conclude in Section~\ref{sec:discussion} with further notable insights and research directions.

% ============================================================================
\section{Preliminaries}
\label{sec:preliminaries}
% ============================================================================

Let $G$ be a (not necessarily simple) \emph{topological graph}, i.e.\ $G$ is a graph drawn on the plane, so that the vertices of $G$ are distinct points in the plane, its edges are Jordan curves joining the corresponding pairs of points, and: %
\begin{inparaenum}[(i)]
\item no edge passes through a vertex different from its endpoints,
\item no edge crosses itself and
\item no two edges meet tangentially.
\end{inparaenum}
Let $\Gamma(G)$ be such a drawing of $G$. The \emph{crossing graph} $\mathcal{X}(G)$ of $G$ has a vertex for each edge of~$G$ and two vertices of $\mathcal{X}(G)$ are connected by an edge if and only if the corresponding edges of $G$ cross in $\Gamma(G)$. A connected component of $\mathcal{X}(G)$ is called \emph{crossing component}. Note that the set of crossing components of $\mathcal{X}(G)$ defines a partition of the edges of $G$. For an edge $e$ of $G$ we denote by $\mathcal{X}(e)$ the crossing component of $\mathcal{X}(G)$ which contains $e$.

An edge $e$ in $\Gamma(G)$ is called a \emph{topological edge} (or simply \emph{edge}, if this is clear in the context). Edge $e$ is called \emph{true-planar}, if it is not crossed by any other edge in $\Gamma(G)$. The set of all true-planar edges of $\Gamma(G)$ forms the so-called \emph{true-planar skeleton} of $\Gamma(G)$, which we denote by $\Pi(G)$. Since $G$ is not necessarily simple, we will assume that $\Gamma(G)$ contains neither \emph{homotopic parallel edges} nor \emph{homotopic self-loops}, that is, both the interior and the exterior regions defined by any self-loop or by any pair of parallel edges contain at least one vertex. For a positive integer $s$, a cycle of length $s$ is called \emph{true-planar $s$-cycle} if it consists of true-planar edges of $\Gamma(G)$. If $e$ is a true-planar edge, then $\mathcal{X}(e)=\{e\}$, while for a chord $e$ of a true-planar $s$-cycle that has no vertices in its interior, it follows that all edges of $\mathcal{X}(e)$ are also chords of this $s$-cycle. Let $\mathcal{F}_s=\{v_1,v_2,\ldots,v_s\}$ be a facial $s$-cycle of $\Pi(G)$ with length $s \geq 3$.  The order of the vertices (and subsequently the order of the edges) of $\mathcal{F}_s$ is determined by a walk around the boundary of $\mathcal{F}_s$ in clockwise direction. Since $\mathcal{F}_s$ is not necessarily simple, a vertex or an edge may appear more than once in this order; see Figure~\ref{fig:non_simple_face}. More in general, a \emph{region} in $\Gamma(G)$ is defined as a closed walk along non-intersecting segments of Jordan curves that are adjacent either at vertices or at crossing points of $\Gamma(G)$. The \emph{interior} and the \emph{exterior} of a connected region are defined as the topological regions to the right and to the left of the walk.

\begin{figure}[tb]
	\centering
	\begin{minipage}[b]{.16\textwidth}
        \centering
        \includegraphics[width=\textwidth,page=1]{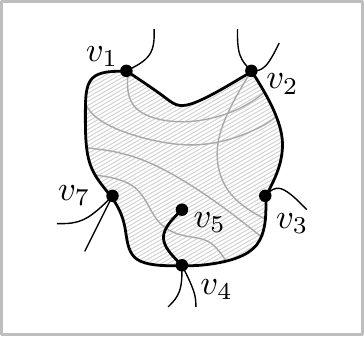}
        \subcaption{~}\label{fig:non_simple_face}
    \end{minipage}
    \begin{minipage}[b]{.16\textwidth}
        \centering
        \includegraphics[width=\textwidth,page=2]{images/pre_cross_twice}
        \subcaption{~}\label{fig:crossing_twice_reverse}
    \end{minipage}
    \begin{minipage}[b]{.16\textwidth}
        \centering
        \includegraphics[width=\textwidth,page=4]{images/pre_cross_twice}
        \subcaption{~}\label{fig:crossing_twice}
    \end{minipage}
	\begin{minipage}[b]{.16\textwidth}
        \centering
        \includegraphics[width=\textwidth,page=3]{images/pre_cross_twice}
        \subcaption{~}\label{fig:crossing_twice_2}
    \end{minipage}
    \begin{minipage}[b]{.16\textwidth}
        \centering
        \includegraphics[width=\textwidth,page=1]{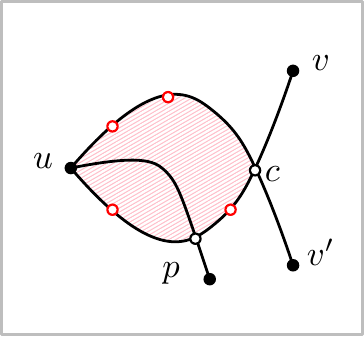}
        \subcaption{~}\label{fig:crossing_adjacent_2}
    \end{minipage}
    \begin{minipage}[b]{.16\textwidth}
        \centering
        \includegraphics[width=\textwidth,page=2]{images/pre_cross_adjacent}
        \subcaption{~}\label{fig:crossing_adjacent}
    \end{minipage}
    \caption{%
    (a)~A non-simple face $\{v_1,\ldots,v_7\}$, where $v_6$ is identified with $v_4$.
    Different configurations used in
    (b--d)~Lemma~\ref{lem:crossing_twice}, and
    (e--f)~Lemma~\ref{lem:crossing_adjacent}.}
    \label{fig:2_planar_polygon_conf}
\end{figure}

Drawing $\Gamma(G)$ is called \emph{$k$-planar} if every edge in $\Gamma(G)$ is crossed at most $k$ times. Accordingly, a graph is called \emph{$k$-planar} if it admits a $k$-planar drawing. An \emph{optimal $k$-planar} graph is a $k$-planar graph with the maximum number of edges. In particular, we consider optimal $2$- and $3$-planar graphs achieving the best-known upper bounds of $5n-10$ and $5.5n-11$ edges. For an optimal $k$-planar graph $G$ on $n$ vertices, a $k$-planar drawing $\Gamma(G)$ of $G$ is called \emph{planar-maximal crossing-minimal} or simply PMCM-drawing, if and only if $\Gamma(G)$ has the maximum number of true-planar edges among all $k$-planar drawings of $G$ and, subject to this restriction, $\Gamma(G)$ has also the minimum number of crossings.

Consider two edges $(u,v)$ and $(u',v')$ that cross at least twice in $\Gamma(G)$. Let $c$ and $c'$ be two crossing points of $(u,v)$ and $(u',v')$ that appear consecutively along $(u,v)$ in this order from $u$ to $v$ (i.e., there is no other crossing point of $(u,v)$ and $(u',v')$ between $c$ and $c'$). W.l.o.g.\ we can assume that $c$ and $c'$ appear in this order along $(u',v')$ from $u'$ to $v'$ as well. In Figures~\ref{fig:crossing_twice_reverse} and \ref{fig:crossing_twice} we have drawn two possible crossing configurations. First we drew edge $(u,v)$ as an arc with $u$ above $v$ and the edge-segment of $(u',v')$ between $u$ and $c$ to the right of $(u,v)$. The edge-segment of $(u',v')$ between $c$ and $c'$, starts at $c$ and ends at $c'$ either from the right (Figure~\ref{fig:crossing_twice_reverse}) or from the left (Figure~\ref{fig:crossing_twice}) of $(u,v)$, yielding the two different crossing configurations.

\begin{lemma}
For $k \in \{2,3\}$, let $\Gamma(G)$ be a PMCM-drawing of an optimal $k$-planar graph $G$ in which two edges $(u,v)$ and $(u',v')$ cross more than once. Let $c$ and $c'$ be two consecutive crossings of $(u,v)$ and $(u',v')$ along $(u,v)$, and let $R_{c,c'}$ be the region defined by the walk along the edge segment of $(u,v)$ from $c$ to $c'$ and the one of $(u',v')$ from $c'$ to $c$. Then, $R_{c,c'}$ has at least one vertex in its interior and one in its exterior.
\label{lem:crossing_twice}
\end{lemma}
\begin{proof}
Consider first the crossing configuration of Figure~\ref{fig:crossing_twice_reverse}. Since $c$ and $c'$ are consecutive along $(u,v)$ and $(u',v')$ does not cross itself, vertex $u'$ lies in the exterior of $R_{c,c'}$, while vertex $v'$ in the interior of $R_{c,c'}$. Hence, the lemma holds. Consider now the crossing configuration of Figure~\ref{fig:crossing_twice}. Since $c$ and $c'$ are consecutive along $(u,v)$, vertices $u'$ and $v'$ are in the exterior of $R_{c,c'}$. Assume now, to the contrary, that $R_{c,c'}$ contains no vertices in its interior. W.l.o.g.~we further assume  that $(u,v)$ and $(u',v')$ is a \emph{minimal crossing pair} in the sense that, $R_{c,c'}$ cannot contain another region $R_{p,p'}$ defined by any other pair of edges that cross twice; for a counterexample see Figure~\ref{fig:crossing_twice_2}. Let $nc(u,v)$ and $nc(u',v')$ be the number of crossings along $(u,v)$ and $(u',v')$ that are between $c$ and $c'$, respectively (red in Figure~\ref{fig:crossing_twice}). Observe that by the ``minimality'' criterion of $(u,v)$ and $(u',v')$ we have $nc(u,v) = nc(u',v')$. We redraw edges $(u,v)$ and $(u',v')$ by exchanging their segments between $c$ and $c'$ and eliminate both crossings $c$ and $c'$ without affecting the $k$-planarity of $G$; see the dotted edges of Figure~\ref{fig:crossing_twice}. This contradicts the crossing minimality of~$\Gamma(G)$.
\end{proof}

\begin{lemma}
For $k \in \{2,3\}$, let $\Gamma(G)$ be a PMCM-drawing of an optimal $k$-planar graph $G$ in which two edges $(u,v)$ and $(u,v')$ incident to a common vertex $u$ cross. Let $c$ be the first crossing of them starting from $u$ and let $R_{c}$ be the region defined by the walk along the edge segment of $(u,v)$ from $u$ to $c$ and the one of $(u,v')$ from $c$ to $u$. Then, $R_{c}$ has at least one vertex in its interior and one in its~exterior.
\label{lem:crossing_adjacent}
\end{lemma}
\begin{proof}
Since $c$ is the first crossing point of $(u,v)$ and $(u,v')$ along $(u,v)$ from $u$, vertex $v'$ is not in the interior of $R_{c}$. If $u \neq v'$, then $v'$  is indeed in the exterior of $R_{c}$. Otherwise, if $u=v'$ and there is no other vertex in the exterior of $R_{c}$, then $(u,v')$ is a homotopic self-loop; a contradiction. Assume now, to the contrary, that $R_{c}$ contains no vertices in its interior. W.l.o.g.~we further assume that $(u,v)$ and $(u,v')$ is a \emph{minimal crossing pair} in the sense that, $R_{c}$ cannot include another region $R_{p}$ defined any other pair of crossing edges incident to a common vertex; for an example see Figure~\ref{fig:crossing_adjacent_2}. Denote by $nc(u,v)$ and $nc(u,v')$ the number of crossings along $(u,v)$ and $(u,v')$ that are between $u$ and $c$, respectively (red drawn in Figure~\ref{fig:crossing_adjacent}). First assume that $nc(u,v) = nc(u,v')$. We proceed by eliminating crossing $c$ without affecting the $k$-planarity of $G$; see the dotted-drawn edges of Figure~\ref{fig:crossing_adjacent}. This contradicts the crossing minimality of $\Gamma(G)$. It remains to consider the case where $nc(u,v) \neq nc(u,v')$. Assume w.l.o.g.~that $nc(u,v) > nc(u,v')$. By the ``minimality''assumption there is an edge $(u'',v'')$ that crosses at least twice edge $(u,v)$. By Lemma~\ref{lem:crossing_twice}, $R_{c}$ is not an empty region; a contradiction.
\end{proof}

In our proofs by contradiction we usually deploy a strategy in which starting from an optimal $2$- or $3$-planar graph $G$, we modify $G$ and its drawing $\Gamma(G)$ by adding and removing elements (vertices or edges) without affecting its $2$- or $3$-planarity. Then, the number of edges in the derived graph forces $G$ to have either fewer or more edges than the ones required by optimality (contradicting the optimality or the $3$-planarity of $G$, resp.). To deploy the strategy, we must ensure that we do not introduce homotopic parallel edges or self-loops, and that we do not violate basic properties of $\Gamma(G)$ (e.g., introduce a self-crossing edge). We next show how to select and draw the newly inserted~elements.

A Jordan curve $[u,v]$ connecting vertex $u$ to $v$ of $G$ is called a \emph{\pe} in drawing $\Gamma(G)$ if and only if $[u,v]$ does not cross itself and is not a homotopic self-loop in $\Gamma(G)$, that is, either $u \neq v$ or $u=v$ and there is at least one vertex in the interior and the exterior of $[u,v]$. Note that $u$ and $v$ are not necessarily adjacent in $G$. However, since each topological edge $(u,v) \in E$ of $G$ is represented by a Jordan curve in $\Gamma(G)$, it follows that edge $(u,v)$ is by definition a \pe of $\Gamma(G)$ among other \pes that possibly exist.
Furthermore, we say that vertices $v_1,v_2,\dots,v_s$ define a \emph{\pec} $\mathcal{C}_s$ in $\Gamma(G)$, if there exist \pes $[v_i,v_{i+1}]$, for $i=1,\dots, s-1$ and \pe $[v_1,v_s]$ of $\Gamma(G)$, which %
\begin{inparaenum}[(i)]
\item do not cross with each other and
\item the walk along the curves between $v_1,v_2,\dots,v_s,v_1$ defines a region in $\Gamma(G)$ that has no vertices in its interior.
\end{inparaenum}
Note that $\mathcal{C}_s$ is not necessarily simple.

\begin{lemma}
For $k \in \{2,3\}$, let $\Gamma(G)$ be a PMCM-drawing of a $k$-planar graph $G$. Let also $\mathcal{C}_s$ be a \pec of length $s$ in $\Gamma(G)$ and assume that $\kappa$ edges of $\Gamma(G)$ are drawn completely in the interior of $\mathcal{C}_s$, while $\lambda$ edges of $\Gamma(G)$ are crossing\footnote{Note that the boundary edges of $\mathcal{C}_s$ are not necessarily present in $\Gamma(G)$.} the boundary of $\mathcal{C}_s$. Also, assume that if one focuses on $\mathcal{C}_s$ of $\Gamma(G)$, then $\mu$ pairwise non-homotopic edges can be drawn as chords completely in the interior of $\mathcal{C}_s$ without deviating $k$-planarity.
\begin{enumerate}[\emph{(}i\emph{)}]
\item \label{prp:nonoptim} If $\mu > \kappa + \lambda$, then $G$ is not optimal.
\item \label{prp:boundary} If $G$ is optimal and $\mu = \kappa + \lambda$, then all boundary edges of $\mathcal{C}_s$ exist\footnote{We say that a Jordan curve $[u,v]$ \emph{exists} in $\Gamma(G)$ if and only if $[u,v]$ is homotopic to an edge in $\Gamma(G)$.} in $\Gamma(G)$.
\end{enumerate}
\label{lem:exchange}
\end{lemma}
\begin{proof}
(\ref*{prp:nonoptim})~If we could replace the $\kappa + \lambda$ edges of $\Gamma(G)$ that are either drawn completely in the interior of $\mathcal{C}_s$ or cross the boundary of $\mathcal{C}_s$ with the $\mu$ ones that one can draw exclusively in the interior of $\mathcal{C}_s$, then the lemma would trivially follow. However, to do so we need to ensure that this operation introduces neither homotopic parallel edges nor homotopic self-loops. Since the edges that we introduce are \pes, it follows that no homotopic self-loops are introduced. We claim that homotopic parallel edges are not introduced either. In fact, if $e$ and $e'$ are two homotopic parallel edges, then both  must be drawn completely in the interior of $\mathcal{C}_s$, which implies that $e$ and $e'$ are both newly-introduced edges; a contradiction, since we introduce $\mu$ pairwise non-homotopic edges. (\ref*{prp:boundary})~In the exchanging scheme that we just described, we drew $\mu$ edges as chords exclusively in the interior of $\mathcal{C}_s$. Of course, one can also draw the boundary edges of $\mathcal{C}_s$, as long as they do not already exist in $\Gamma(G)$. Since $G$ is optimal, these edges must exist in $\Gamma(G)$.
\end{proof}

\begin{figure}[t]
	\centering
	\begin{minipage}[b]{.18\textwidth}
        \centering
        \includegraphics[width=\textwidth,page=1]{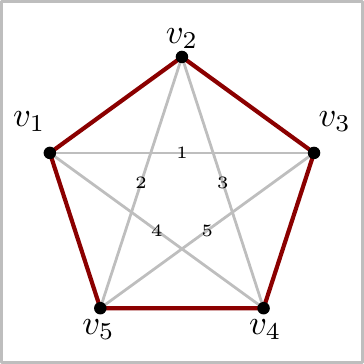}
        \subcaption{~}\label{fig:2_planar_5gon}
    \end{minipage}
    \begin{minipage}[b]{.18\textwidth}
        \centering
        \includegraphics[width=\textwidth,page=2]{images/pre_chords}
        \subcaption{~}\label{fig:2_planar_6gon}
    \end{minipage}
	\begin{minipage}[b]{.18\textwidth}
        \centering
        \includegraphics[width=\textwidth,page=3]{images/pre_chords}
        \subcaption{~}\label{fig:3_planar_6gon}
    \end{minipage}
    \begin{minipage}[b]{.18\textwidth}
        \centering
        \includegraphics[width=\textwidth,page=1]{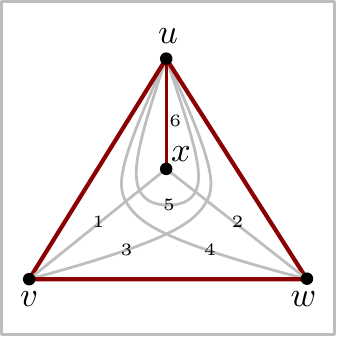}
        \subcaption{~}\label{fig:2_planar_triangle}
    \end{minipage}
    \caption{%
    (a--c)~A \pec $\mathcal{C}_s$ with
	(a)~$s=5$ and five chords with two crossings each,
    (b)~$s=6$ and six chords with at most two crossings each, and
    (c)~$s=6$ and eight chords with at most three crossings each.
    (d)~Configuration used in the proof of Property~\ref{prp:2planar_triangle}.}
    \label{fig:chord_conf}
\end{figure}

Note that in Lemma~\ref{lem:exchange} the $\kappa$ edges that are drawn completely in the interior of the \pec $\mathcal{C}_s$ and the $\lambda$ edges that cross its boundary, are the only edges that have at least one edge-segment within $\mathcal{C}_s$. This means that we can compute $\kappa+\lambda$ by counting the edges that have at least one edge-segment within $\mathcal{C}_s$. In the following sections, there will be some standard cases where we apply Lemma~\ref{lem:exchange}. In most of them, a \pec $\mathcal{C}_s$ on five or six vertices is involved, that is, $5 \leq s \leq 6$. If $s=5$, then one can draw five chords in the interior of $\mathcal{C}_s$ without affecting its $2$- or $3$-planarity; see Figure~\ref{fig:2_planar_5gon}. If $s=6$, then one can draw either six or eight chords in the interior of $\mathcal{C}_s$ without affecting its $2$- or $3$-planarity, respectively; see Figures~\ref{fig:2_planar_6gon} and~\ref{fig:3_planar_6gon}.

% ============================================================================
\section{Properties of optimal 2- and 3-planar graphs}
\label{sec:properties}
% ============================================================================

In this section, we investigate properties of optimal $2$- and $3$-planar graphs.We prove that a PMCM-drawing $\Gamma(G)$ of an optimal $2$- or $3$-planar graph $G$ can contain neither true-planar cycles of a certain length nor a pair of edges that cross twice. We use these properties to show that $\Gamma(G)$ is \emph{quasi-planar}, i.e.\ it contains no $3$ pairwise crossing edges. First, we give the following definition. Let $R$ be a simple closed region that contains at least one vertex of $G$ in its interior and one in its exterior. Let $H_1$ ($H_2$) be the subgraph of $G$ whose vertices and edges are drawn entirely in the interior (exterior) of $R$. Note that $H_1$ ($H_2$) is not necessarily an induced subgraph of $G$, since there could be edges that exit and enter $R$. We refer to $H_1$ and $H_2$ as the \emph{compact subgraphs} of $\Gamma(G)$ defined by $R$. The following lemma, used in the proofs for several properties of optimal $2$- and $3$-planar graphs, bounds the number of edges in any compact subgraph of $\Gamma(G)$.

\begin{property}
Let $\Gamma(G)$ be a drawing of an optimal $2$- or $3$-planar graph $G$ and let $H$ be a compact subgraph of $\Gamma(G)$ on $n'$ vertices that is defined by a closed region $R$. If $n'\geq 2$,  $H$ has at most $5n'-6$ edges if $G$ is optimal $2$-planar, and at most $5.5n' - 6.5$ edges if $G$ is optimal $3$-planar. Furthermore, there exists at least one edge of $G$ crossing the boundary of $R$ in $\Gamma(G)$.
\label{prp:connected}
\end{property}
\begin{proof}
We prove this property for the class of $3$-planar graphs; the proof for the class of $2$-planar graphs is analogous.
So, let $\Gamma(G)$ be a drawing of an optimal $3$-planar graph $G=(V,E)$ with $n$ vertices and $m$ edges. Let $H_1$ and $H_2$ be two compact subgraphs of $\Gamma(G)$ defined by a closed region $R$.
For $i=1,2$ let $n_i$ and $m_i$ be the number of vertices and edges of $H_i$.
Suppose that $n_1\geq 2$. In the absence of $\Gamma(H_2)$, drawing $\Gamma(H_1)$ might contain homotopic parallel edges or self-loops. To overcome this problem, we subdivide an edge-segment of the unbounded region of $\Gamma(H_1)$ by adding one vertex.\footnote{One can view this process as replacing $\Gamma(H_2)$ with a single vertex; thus no homotopic parallel edges exist in $\Gamma(H_1)$. Then we move this vertex towards the edge-segment we want to subdivide until it touches it.} The derived graph, say $H'_1$, has $n'_1=n_1+1$ vertices and $m'_1=m_1+1$ edges. Since $H'_1$ has no homotopic parallel edges or self-loops and $n'_1 \geq 3$, it follows that $m'_1 \leq 5.5n'_1-11$, which gives $m_1\leq 5.5n_1-6.5$.

For the second part, assume for the sake of contradiction that no edge of $G$ crosses the boundary of $R$. This implies that $m=m_1+m_2$. We consider first the case where $n_1,n_2 \geq 2$. By the above we have that $m_1\leq 5.5n_1-6.5$and  $m_2 \leq 5.5n_2 - 6.5$. Since $n=n_1+n_2$ and $m=m_1+m_2$, it follows that $m \leq 5.5n -13$; a contradiction to the optimality of~$G$. Since a graph consisting only of two non-adjacent vertices cannot be optimal, it remains to consider the case where either $n_1=1$ or $n_2=1$. W.l.o.g.\ assume that $n_1=1$. Since $n_2 \geq 2$, it follows that $m_2 \leq 5.5n_2 - 6.5$, which implies $m \leq 5.5n - 12$; a contradiction to the optimality of $G$.
\end{proof}

For two compact subgraphs $H_1$ and $H_2$ defined by a closed region $R$, Property~\ref{prp:connected} implies that the drawings of $H_1$ and $H_2$ cannot be ``separable''. In other words, either there exists an edge connecting a vertex of $H_1$ with a vertex of $H_2$, or there exists a pair of edges, one connecting vertices of $H_1$ and the other vertices of $H_2$, that cross in the drawing $\Gamma(G)$.

\begin{property}
In a PMCM-drawing $\Gamma(G)$ of an optimal $2$-planar graph $G$ there is no empty true-planar cycle of length three.
\label{prp:2planar_triangle}
\end{property}
\begin{proof}
Assume to the contrary that there exists an empty true-planar $3$-cycle $\mathcal{C}$ in $\Gamma(G)$ on vertices $u$, $v$ and $w$. Since $G$ is connected and since all edges of $\mathcal{C}$ are true-planar, there is neither a vertex nor an edge-segment in $\mathcal{C}$, i.e., $\mathcal{C}$ is a chordless facial cycle of $\Pi(G)$. This allows us to add a vertex $x$ in its interior and connect $x$ to vertex $u$ by a true-planar edge. Now vertices $u$, $x$, $u$, $w$ and $v$ define a \pec of length five, and we can draw five chords in its interior without violating $2$-planarity and without introducing homotopic parallel edges or self-loops; refer to Figure~\ref{fig:2_planar_triangle}. The derived graph $G'$ has one more vertex than $G$ and six more edges. Hence, if $n$ and $m$ are the number of vertices and edges of $G$ respectively, then $G'$ has $n'=n+1$ vertices and $m'=m+6$ edges. Then $m'=5n'-9$, which implies that $G'$ has more edges than allowed; a contradiction.
\end{proof}

\begin{property}
The number of vertices of an optimal $3$-planar graph $G$ is even.
\label{prp:3planar_even_order}
\end{property}
\begin{proof}
Follows directly from the density bound of $5.5n-11$ of $G$.
\end{proof}

\begin{property}
A PMCM-drawing $\Gamma(G)$ of an optimal $3$-planar graph $G$ has no true-planar cycle of odd length.
\label{prp:3planar_odd_cycle}
\end{property}
\begin{proof}
Let $s \geq 1$ be an odd number and assume to the contrary that there exists a
true-planar $s$-cycle $\mathcal{C}$ in $\Gamma(G)$.  Denote by $G_1$ ($G_2$,
respectively) the subgraph of $G$ induced by the
vertices of $\mathcal{C}$ and the
vertices of $G$ that are in the interior (exterior, respectively) of $\mathcal{C}$ in
$\Gamma(G)$ without the chords of $\mathcal{C}$ that are in the exterior
(interior, respectively) of $\mathcal{C}$ in $\Gamma(G)$. For $i=1,2$, observe that
$G_i$ contains a copy of $\mathcal{C}$. Let $n_i$ and $m_i$ be the number of
vertices and edges of $G_i$ that do not belong to $\mathcal{C}$. Based on graph
$G_i$, we construct graph $G_i'$ by employing two copies of $G_i$ that share
cycle $\mathcal{C}$. Observe that $G_i'$ is $3$-planar, because one copy of
$G_i$ can be embedded in the interior of $\mathcal{C}$, while the other one in
its exterior. Hence, in this embedding, there exist neither homotopic self-loops
nor homotopic parallel edges. Let $n_i'$ and $m_i'$ be the number of vertices
and edges of $G_i'$ that do not belong to $\mathcal{C}$. If $G$ has $n$ vertices
and $m$ edges, then by construction the following equalities hold:
\begin{inparaenum}[(i)]
\item \label{il:1} $n_i'= 2n_i+s$,
\item \label{il:2} $m_i'= 2m_i+s$,
\item \label{il:3} $n = n_1+n_2+s$, and
\item \label{il:4} $m = m_1+m_2+s$.
\end{inparaenum}

We now claim that $n_i' \geq 3$. When $s \geq 3$ the claim clearly holds. Otherwise (i.e., $s=1$), cycle $\mathcal{C}$ is degenerated to a self-loop which must contain at least one vertex in its interior and its exterior. Hence, the claim follows. Property~\ref{prp:3planar_even_order} in conjunction with Eq.(\ref{il:1}) implies that $G_i'$ is not optimal, that is, $m_i'<5.5n_i'-11$. Hence, by Eq.(\ref{il:2}) it follows that $2m_i+s < 5.5(2n_i+s)-11$. Summing up over $i$, we obtain that $2(m_1+m_2+s) < 5.5 (2n_1+2n_2+2s)-22$. Finally, from Eq.(\ref{il:3}) and Eq.(\ref{il:4}) we conclude that $m<5.5n-11$; a contradiction to the optimality of $G$.
\end{proof}

\begin{property}
In a PMCM-drawing $\Gamma(G)$ of an optimal $2$-planar graph $G$ there is no pair of edges that cross twice with each other.
\label{prp:2_planar_cross_twice}
\end{property}
\begin{proof}
Assume to the contrary that $(u,u')$ and $(v,v')$ cross twice in $\Gamma(G)$ at points $c$ and $c'$. By $2$-planarity no other edge of $\Gamma(G)$ crosses $(u,u')$ and $(v,v')$. Let $R_{c,c'}$ be the region defined by the walk along the edge segment of $(u,u')$ between $c$ and $c'$ and the edge segment of $(v,v')$ between $c'$ and $c$. As mentioned in the proof of Lemma~\ref{lem:crossing_twice}, there exist two crossing configurations for $(u,u')$ and $(v,v')$; see Figures~\ref{fig:crossing_twice_reverse} and~\ref{fig:crossing_twice}. In the crossing configuration of Figure~\ref{fig:crossing_twice_reverse}, vertices $v$ and $v'$ are in the interior of $R_{c,c'}$, while vertices $u$ and $u'$ in its exterior. Hence, $u\neq v$ and $u'\neq v'$ hold. We redraw $(u,u')$ and $(v,v')$ by exchanging the middle segments between $c$ and $c'$ and eliminate both crossings $c$ and $c'$ without affecting $2$-planarity; see the dotted edges of Figure~\ref{fig:crossing_twice_reverse}. Note that since $u\neq v$ and $u'\neq v'$ the two edges cannot be homotopic self-loops. Also, no homotopic parallel edges are introduced, since this would imply that at least one of the two edges already exists in $\Gamma(G)$ violating $2$-planarity. Now consider the crossing configuration of Figure~\ref{fig:crossing_twice}. By Lemma~\ref{lem:crossing_twice}, $R_{c,c'}$ has at least one vertex in its interior. By $2$-planarity we have that no edge of $G$ crosses the boundary of $R_{c,c'}$; a contradiction to Property~\ref{prp:connected}.
\end{proof}

\begin{property}
In a PMCM-drawing $\Gamma(G)$ of an optimal $3$-planar graph $G$ there is no pair of edges that cross more than once with each other.
\label{prp:3_planar_cross_twice}
\end{property}
\begin{proof}
We have already noted that a pair of edges cannot cross more than twice in $\Gamma(G)$. Assume to the contrary that two edges $(u,v)$ and $(u',v')$ of $G$ cross (exactly) twice in $\Gamma(G)$. Figures~\ref{fig:3_planar_cross_twice_general_reverse} and~\ref{fig:3_planar_cross_twice_general} illustrate the two possible different crossing configurations. Let $c$ and $c'$ be their crossing points. By Lemma~\ref{lem:crossing_twice} it follows that the region $R_{c,c'}$ that is defined by the walk along the the edge segment of $(u,v)$ between $c$ and $c'$ and the edge segment of $(u',v')$ between $c'$ and $c$ has at least one vertex in its interior. Let $G_{c,c'}$ be the subgraph of $G$ that is drawn completely in the interior of $R_{c,c'}$ in $\Gamma(G)$. By $3$-planarity, there exist at most two edges $e$ and $e'$ that cross $(u,v)$ and $(u',v')$ respectively.

In both crossing configurations we proceed to define two Jordan curves $[u,u']_1$ and $[u,u']_2$ in $\Gamma(G)$ with endpoints $u$ and $u'$, so that their union contains only in its interior the vertices of $G_{c,c'}$; see Figures~\ref{fig:3_planar_cross_twice_general_reverse} and~\ref{fig:3_planar_cross_twice_general}. Curve $[u,u']_1$ emanates from vertex $u$, follows edge $(u,v)$ up to point $c$ and ends at vertex $u'$ by following edge $(u',v')$. Curve $[u,u']_2$ emanates from vertex $u'$, follows edge $(u',v')$ up to point $c$, follows edge $(u,v)$ up to point $c'$, follows edge $(u',v')$ up to point $c$ and ends at vertex $u$ by following edge $(u,v)$.

We now claim that both curves $[u,u']_1$ and $[u,u']_2$ are \pes. By definition, our claim holds when $u \neq u'$. Assume now that $u=u'$. Let $R_{c}$ be the region defined by the walk along the edge-segment of $(u',v')$ from $u'$ to $c$ and the edge-segment of $(u,v)$ from $c$ to $u$ (where $u=u'$). By Lemma~\ref{lem:crossing_adjacent} $R_{c}$ has at least one vertex in its interior and at least one vertex in its exterior. This implies that the first of our curves, i.e. $[u,u']_1$, which encloses region $R_c$ is a \pe.

Now, assume to the contrary that $[u,u']_2$ is not a \pe. Then $u=u'$. Let $R_{c'}$ be the region defined by the walk along the edge-segment of $(u',v')$ from $u'$ to $c$, the edge-segment of $(u,v)$ from $c$ to $c'$, the edge-segment of $(u',v')$ from $c'$ to $c$ and the edge-segment of $(u,v)$ from $c$ to $u$ (where $u=u'$). Since $G_{c,c'}$ lies in the interior of $R_{c'}$ and $[u,u']_2$ is not a \pe, region $R_{c'}$ has no vertices in its exterior; refer to Figure~\ref{fig:3_planar_cross_twice_special_reverse}. Note that in Figure~\ref{fig:3_planar_cross_twice_special} we illustrate the same case assuming $u=u'=v=v'$. By Property~\ref{prp:3planar_odd_cycle} \pe $[u,u']_1$ must be crossed (as otherwise it is a true-planar self-loop in $\Gamma(G)$). This implies that there exists at least one edge that crosses $[u,u']_1$. This edge must also cross $(u,v)$ or $(u',v')$ and is therefore either edge $e$ or edge $e'$. Suppose w.l.o.g.~that $[u,u']_1$ is edge $e$; see Figure~\ref{fig:3_planar_cross_twice_special_reverse}. Let $p$ be the crossing point of $e$ and $(u,v)$. Now edge $(u,v)$ has exactly three crossings. We redraw $(u,v)$ and $(u',v')$ by exchanging their edge-segments between their common endpoint $u$ and their first crossing $c$, so as to eliminate $c$. Let $[u,v]$ and $[u',v']$ be the new curves in $\Gamma(G)$. Since $G$ is crossing minimal, it follows that at least one of $[u,v]$ or $[u',v']$ must be homotopic parallel to an existing edge in $\Gamma(G)$. Since $(u,v)$ has already three crossings in $\Gamma(G)$, \pe $[u',v']$ cannot exist in $\Gamma(G)$, as otherwise it would introduce a fourth crossing on $(u,v)$. Hence, \pe $[u,v]$ must exist in $\Gamma(G)$ and this is edge $e'$. Now we focus on edge $e$. Edge $e$ has an endpoint in the interior of $R_c$ and crosses $[u,u']_2$. However, since $R_{c'}$ has no vertices in its exterior, and edges $(u,v)$ and $(u',v')$ have already three crossings, edge $e$ must end at vertex $u=u'$. In this case, edges $e$ and $(u,v)$ have $u$ as a common endpoint and cross at point $p$. Hence, region $R_p$ defined by the walk along the edge segment of $(u,v)$ from $u$ to $p$ and the edge segment of $e$ from $p$ to $u$ contains at least one vertex in its interior. However, $R_p$ is contained in the exterior of $R_{c'}$, and therefore there exists at least one vertex in the exterior of $R_{c'}$, which is a contradiction. Hence, $[u,u']_2$ is a \pe.

\begin{figure}[t]
	\centering
    \begin{minipage}[b]{.18\textwidth}
        \centering
		\includegraphics[width=\textwidth,page=2]{images/prop_2planar_cross_twice}
        \subcaption{~}\label{fig:3_planar_cross_twice_general_reverse}
    \end{minipage}
     \begin{minipage}[b]{.18\textwidth}
        \centering
		\includegraphics[width=\textwidth,page=1]{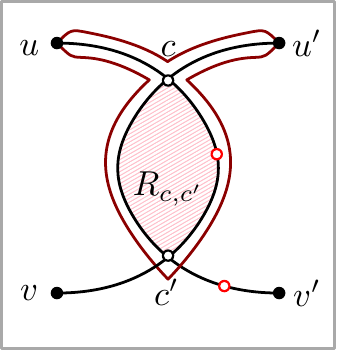}
        \subcaption{~}\label{fig:3_planar_cross_twice_general}
    \end{minipage}
	\begin{minipage}[b]{.18\textwidth}
        \centering
		\includegraphics[width=\textwidth,page=3]{images/prop_2planar_cross_twice}
        \subcaption{~}\label{fig:3_planar_cross_twice_special_reverse}
    \end{minipage}
	\begin{minipage}[b]{.18\textwidth}
        \centering
        \includegraphics[width=\textwidth,page=2]{images/prop_3planar_cross_twice}
        \subcaption{~}\label{fig:3_planar_cross_twice_special}
    \end{minipage}
    \caption{%
    Configurations used in Property~\ref{prp:3_planar_cross_twice}.
    \label{fig:3_planar_one_crossing_2}}
\end{figure}

We proceed by removing from $\Gamma(G)$ all vertices and edges of $G_{c,c'}$, edges $e$, $(u,v)$, $(u',v')$ as well as the edge that crosses $(u',v')$, if any. Then, the cycle formed by \pes $[u,u']_1$ and $[u,u']_2$ becomes empty and this allows us to follow an approach similar to the one described in the proof of Lemma~\ref{lem:exchange}. More precisely, we add in the interior of this \pec two vertices $x$ and $y$, such that $u$, $x$ and $y$ form a path (in this order) that is completely drawn in its interior. The union of this path with $[u,u']_1$ and $[u,u']_2$ defines in the derived drawing a new (non-simple) \pec of length six. In its interior one can embed $8$ additional edges as in Figure~\ref{fig:3_planar_6gon}. Summarizing, if $G_{c,c'}$ has $n_{c,c'}$ vertices and $m_{c,c'}$ edges, we removed from $G$ exactly $n_{c,c'}$ vertices and at most $m_{c,c'}+4$ edges and this allowed us to introduce two new vertices and $10$ edges without affecting $3$-planarity. Let $G'$ be the derived $3$-planar graph. The fact that $G'$ contains neither homotopic parallel edges nor homotopic self-loops can be argued as in the proof of Lemma~\ref{lem:exchange}.(\ref{prp:nonoptim}). If $G$ has $n$ vertices and $m$ edges, then  $G'$ has $n' = n -n_{c,c'}+2$ vertices and $m'$ edges, where $m' \geq m -m_{c,c'}+ 6$ edges. We distinguish two cases depending on whether $G_{c,c'}$ has one or more vertices. If $n_{c,c'}=1$, then $m_{c,c'}=0$. Also, $G'$ has exactly one more vertex than $G$. Since $G$ is optimal, by Property~\ref{prp:3planar_even_order} it follows that $G'$ cannot be optimal. Hence, $m' < 5.5n' - 11$, which implies that $m < 5.5n - 11.5$; a contradiction to the density of $G$. On the other hand if $n_{c,c'}\geq 2$, by Property~\ref{prp:connected} we have that $m_{c,c'}\leq 5.5n_{c,c'}-6.5$, as $G_{c,c'}$ is a compact subgraph of $\Gamma(G)$ defined by $R_{c,c'}$. This gives $m'\geq 5.5n'-9.5$, that is $G'$ has more edges than allowed; a clear contradiction.
\end{proof}

Now assume that $\Gamma(G)$ contains three mutually crossing edges $(u,v)$, $(u',v')$ and $(u'',v'')$. In Figures~\ref{fig:quasi_1}--\ref{fig:quasi_4} we have drawn four possible crossing configurations. First, we drew $(u,v)$ and $(u',v')$ w.l.o.g.\ as vertical and horizontal line-segments that cross at point $c$. Then, we placed vertex $u''$ and drew the first segment of its edge crossing w.l.o.g.~the edge-segment of $(u',v')$ between $u'$ and $c$ at point $c'$ from above. So the middle segment of $(u'',v'')$ starts at $c'$ and has to end at edge $(u,v)$, either from left or right, and either in the lower or in the upper segment. This gives rise to the four configurations demonstrated in Figures~\ref{fig:quasi_1}--\ref{fig:quasi_4}, which we examine in more details in the following. Note that the endpoints of the three edges are not necessarily distinct (e.g., in Figure~\ref{fig:quasi_1_1} we illustrate the case where $u=u''$ and $v'=v''$ for the crossing configuration of Figure~\ref{fig:quasi_1}). For each crossing configuration, one can draw curves connecting the endpoints of $(u,v)$, $(u',v')$ and $(u'',v'')$ (red colored in Figures~\ref{fig:quasi_1}--\ref{fig:quasi_4}), which define a region that has no vertices in its interior. This region fully surrounds $(u,v)$ and $(u',v')$ and the two segments of $(u'',v'')$ that are incident to vertices $u''$ and $v''$.

\begin{figure}[tb]
    \centering
	\begin{minipage}[b]{.19\textwidth}
        \centering
        \includegraphics[width=\textwidth,page=1]{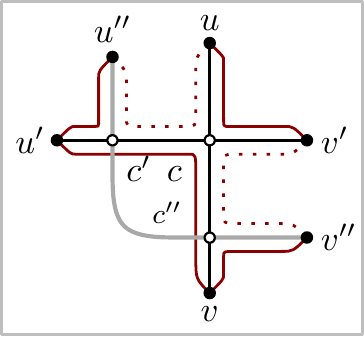}
        \subcaption{~}\label{fig:quasi_1}
	\end{minipage}
    \begin{minipage}[b]{.19\textwidth}
        \centering
        \includegraphics[width=\textwidth,page=2]{images/prop_quasi_planar}
        \subcaption{~}\label{fig:quasi_2}
    \end{minipage}
    \begin{minipage}[b]{.19\textwidth}
        \centering
        \includegraphics[width=\textwidth,page=3]{images/prop_quasi_planar}
        \subcaption{~}\label{fig:quasi_3}
    \end{minipage}
    \begin{minipage}[b]{.19\textwidth}
        \centering
        \includegraphics[width=\textwidth,page=4]{images/prop_quasi_planar}
        \subcaption{~}\label{fig:quasi_4}
    \end{minipage}
    \begin{minipage}[b]{.19\textwidth}
        \centering
        \includegraphics[width=\textwidth,page=5]{images/prop_quasi_planar}
        \subcaption{~}\label{fig:quasi_1_1}
    \end{minipage}
    \caption{%
    Crossing configurations for three mutually crossing edges.
    Potential edges are drawn solid red.
    Jordan curves that can either be potential edges or homotopic self-loops are drawn dotted red.}
    \label{fig:prp_quasi}
\end{figure}

\begin{claim}
Each of the crossing configurations of Figures~\ref{fig:quasi_2}-\ref{fig:quasi_4} induces at least $5$ \pes.
\label{clm:1}
\end{claim}
\begin{proof}
Observe that all solid-drawn red curves of Figures~\ref{fig:quasi_2}--\ref{fig:quasi_4} are indeed \pes: If for example $[u',u'']$ of Figure~\ref{fig:quasi_2} is not a \pe, then $u'=u''$ and $[u',u'']$ is a self-loop with no vertex either in its interior or in its exterior; a contradiction to Lemma~\ref{lem:crossing_adjacent}.
\end{proof}

\begin{claim}
The crossing configuration of Figure~\ref{fig:quasi_1} induces at least four \pes.
\label{clm:2}
\end{claim}
\begin{proof}
As in Claim~\ref{clm:1} we can prove that $[u',u'']$, $[u,v']$, $[u',v]$ and $[v,v'']$ are \pes.
\end{proof}

\begin{crl}
The configuration of Figure~\ref{fig:quasi_1} induces a \pec~$\mathcal{C}$ of length $\geq 4$. Each of the configurations of Figures~\ref{fig:quasi_2}--\ref{fig:quasi_4} induces a \pec~$\mathcal{C}$ of length $\geq 5$.
\label{crl:crl}
\end{crl}

\begin{claim}
In the case where the crossing configuration of Figure~\ref{fig:quasi_1} induces exactly four \pes, there exists at least one vertex in the interior of region $\mathcal{T}$ defined by the walk along the edge segment of $(u,v)$ between $c$ and $c''$, the edge segment of $(u'',v'')$ between $c''$ and $c'$ and the edge segment of $(u',v')$ between $c'$ and $c$.
\label{clm:3}
\end{claim}
\begin{proof}
By Claim~\ref{clm:2}, $[u,u'']$, and $[v',v'']$ must be homotopic self-loops; see Figure~\ref{fig:quasi_1_1}. In this case, edges $(u,v)$ and $(u'',v'')$ are incident to a common vertex, namely $u=u''$ and cross. By Lemma~\ref{lem:crossing_adjacent} region $R_{c''}$ (red-shaded in Figure~\ref{fig:quasi_1_1}) has at least one vertex in its interior. Since $R_{c''}$ is the union of the interior of $\mathcal{T}$ and the homotopic self-loop $[u,u'']$, $\mathcal{T}$ contains  at least one vertex in its interior.
\end{proof}

\begin{property}
A PMCM-drawing $\Gamma(G)$ of an optimal $2$-planar graph $G$ is quasi-planar.
\label{prp:2_planar_quasi}
\end{property}
\begin{proof}
Assume to the contrary that there exist three mutually crossing edges $(u,v)$, $(u',v')$ and $(u'',v'')$ in $\Gamma(G)$; see Figure~\ref{fig:prp_quasi}. By Corollary~\ref{crl:crl}, there is a \pec $\mathcal{C}$ of length at least $4$. By $2$-planarity, there is no other edge crossing $(u,v)$, $(u',v')$ or $(u'',v'')$. Hence, the only edges that are drawn in the interior of $\mathcal{C}$ are $(u,v)$ and $(u',v')$, while $(u'',v'')$ is the only edge that crosses the boundary of $\mathcal{C}$.

First, consider the case where $\mathcal{C}$ is of length $\geq 5$. Since we can draw at least five chords completely in the interior of $\mathcal{C}$ as in Figure~\ref{fig:2_planar_5gon} or \ref{fig:2_planar_6gon} without violating its $2$-planarity, it follows by Lemma~\ref{lem:exchange}.(\ref{prp:nonoptim}) (for $\kappa+\lambda=3$ and $\mu \geq 5$) that $G$ is not optimal; a contradiction. Finally, consider the case where $\mathcal{C}$ is of length four. In this case, we have the crossing configuration of Figure~\ref{fig:quasi_1}. By Claim~\ref{clm:3} there is at least one vertex in the interior of region $\mathcal{T}$. More in general, let $G_{\mathcal{T}}$ be the compact subgraph of $G$ that is completely drawn in the interior of region $\mathcal{T}$. Since edges $(u,v)$, $(u',v')$ and $(u'',v'')$ have already two crossings, it follows that no edge of $G$ crosses the boundary of $\mathcal{T}$; a contradiction to Property~\ref{prp:connected}.
\end{proof}

\begin{property}
A PMCM-drawing $\Gamma(G)$ of an optimal $3$-planar graph $G$ is quasi-planar.
\label{prp:3_planar_quasi}
\end{property}
\begin{proof}
As in the case of $2$-planar optimal graphs, assume that there exist three mutually crossing edges $(u,v)$, $(u',v')$ and $(u'',v'')$ in $\Gamma(G)$. By Corollary~\ref{crl:crl}, there is always a \pec $\mathcal{C}$ of length at least $4$. Since $(u,v)$, $(u',v')$ and $(u'',v'')$ have already two crossings each, there exist at most three other edges that cross $(u,v)$, $(u',v')$ or $(u'',v'')$. Hence, the only edges that are drawn in the interior of $\mathcal{C}$ are $(u,v)$ and $(u',v')$, while $(u'',v'')$ and at most three other edges of $\Gamma(G)$ cross the boundary of $\mathcal{C}$.  We distinguish three cases depending on whether $\mathcal{C}$ has length $6$, $5$ or $4$.

Consider first the case where $\mathcal{C}$ has length six. Since we can draw eight chords completely in the interior of $\mathcal{C}$ as in Figure~\ref{fig:3_planar_6gon} without deviating $3$-planarity, it follows by Lemma~\ref{lem:exchange}.(\ref{prp:nonoptim}) (for $\kappa+\lambda=6$ and $\mu=8$) that $G$ is not optimal; a contradiction.

Consider now the case where $\mathcal{C}$ has length five. We claim that at least one boundary edge of $\mathcal{C}$ does not exist in $\Gamma(G)$. In order to prove the claim, we consider the four crossing configurations of Figure~\ref{fig:prp_app_quasi} separately. In Figure~\ref{fig:quasi_app_1}, if \pe $[u',v]$ is an edge in $\Gamma(G)$, then it crosses twice $(u'',v'')$, contradicting Property~\ref{prp:3_planar_cross_twice}. For Figures~\ref{fig:quasi_app_2}--\ref{fig:quasi_app_4}, if all red drawn curves belong to $\Gamma(G)$, then $(u'',v'')$ crosses $(u,v)$, $(u',v')$ and at least two of the boundary edges of $\mathcal{C}$, violating $3$-planarity. Hence, our claim follows. We proceed by removing edges $(u,v)$, $(u',v')$ and $(u'',v'')$ and any other edge crossing the boundary of $\mathcal{C}$ from $\Gamma(G)$, and we add five chords in the interior of $\mathcal{C}$, along with one ``missing'' boundary edge of $\mathcal{C}$. Let $G'$ be the derived graph. Note that, we removed at most six edges and added at least six. This implies that $G'$ is also optimal. However, $\mathcal{C}$ is a true-planar $5$-cycle in the drawing of $G'$, contradicting Property~\ref{prp:3planar_odd_cycle}.

\begin{figure}[tb]
	\centering
	\begin{minipage}[b]{.16\textwidth}
        \centering
        \includegraphics[width=\textwidth,page=1]{images/prop_quasi_planar}
        \subcaption{~}\label{fig:quasi_app_1}
	\end{minipage}
    \begin{minipage}[b]{.16\textwidth}
        \centering
        \includegraphics[width=\textwidth,page=2]{images/prop_quasi_planar}
        \subcaption{~}\label{fig:quasi_app_2}
    \end{minipage}
    \begin{minipage}[b]{.16\textwidth}
        \centering
        \includegraphics[width=\textwidth,page=3]{images/prop_quasi_planar}
        \subcaption{~}\label{fig:quasi_app_3}
    \end{minipage}
    \begin{minipage}[b]{.16\textwidth}
        \centering
        \includegraphics[width=\textwidth,page=4]{images/prop_quasi_planar}
        \subcaption{~}\label{fig:quasi_app_4}
    \end{minipage}
    \begin{minipage}[b]{.16\textwidth}
        \centering
        \includegraphics[width=\textwidth,page=5]{images/prop_quasi_planar}
        \subcaption{~}\label{fig:quasi_app_1_1}
    \end{minipage}
    \begin{minipage}[b]{.16\textwidth}
        \centering
        \includegraphics[width=\textwidth,page=6]{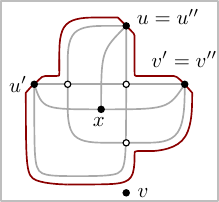}
        \subcaption{~}\label{fig:3_planar_quasi}
    \end{minipage}
    \caption{%
    Crossing configurations for three mutually crossing edges.
    Potential edges are drawn solid red.
    Jordan curves that can either be potential edges or homotopic self-loops are drawn dotted red.}
    \label{fig:prp_app_quasi}
\end{figure}

It remains to consider the case where $\mathcal{C}$ is of length four. By Claim~\ref{clm:3} there is at least one vertex in the interior of region $\mathcal{T}$. As in the proof of Property~\ref{prp:2_planar_quasi}, we denote by $G_{\mathcal{T}}$ the subgraph of $G$ completely drawn in region $\mathcal{T}$. $G_{\mathcal{T}}$ is a compact subgraph of $\Gamma(G)$ and by Property~\ref{prp:connected}, it follows that if $G_{\mathcal{T}}$ has $n_{\mathcal{T}}\geq 2$ vertices, then it has $m_{\mathcal{T}}\leq 5.5n_{\mathcal{T}}-6.5$ edges (note that if $n_{\mathcal{T}}=1$, then $m_{\mathcal{T}}=0$). We replace $G_{\mathcal{T}}$ with one vertex, say $x$, we keep edges $(u,v)$, $(u',v')$ and $(u'',v'')$ and remove any edge crossing $(u,v)$, $(u',v')$ or $(u'',v'')$ in $\Gamma(G)$. We redraw the edge-segment of $(u,v)$ incident to $v$ so as to be incident to $u'$ (without introducing new  crossings). Finally, we add edges $(x,u)$, $(x,u')$ and $(x,v')$; see Figure~\ref{fig:3_planar_quasi}. The derived graph $G'$ has $n'=n-n_{\mathcal{T}}+1$ vertices and at least $m'\geq m-m_{\mathcal{T}}$ edges, where $n$ and $m$ are the number of vertices and edges of $G$. For $n_{\mathcal{T}}\geq 2$, we have that $m'\geq 5.5n'-10$, i.e., $G'$ has more edges than allowed. In the case where $n_{\mathcal{T}}=1$ and $m_{\mathcal{T}}=0$, it follows that $G'$ has the same number of edges as $G$ and is therefore optimal. However, \pes $[u,v']$, $[u',u'']$ and $[u',v'']$ can be added in $\Gamma(G')$ (if not present) forming thus a true-planar $3$-cycle; a contradiction to Property~\ref{prp:3planar_odd_cycle}.
\end{proof}

We next present a refinement of the notion of \pes. In particular, we focus on two main categories of \pes that we will heavily use in Sections~\ref{sec:2planar} and \ref{sec:3planar}. Consider a pair of vertices $u$ and $v$ of $G$ that are not necessarily distinct. We say that $u$ and $v$ form a \emph{corner pair} if and only if an edge $(u,u')$ crosses an edge $(v,v')$ for some $u'$ and $v'$ in $\Gamma(G)$; see Figure~\ref{fig:corner_pair}. Let $c$ be the crossing point of $(u,u')$ and $(v,v')$. Then, any Jordan curve $[u,v]$ joining vertices $u$ and $v$ induces a region $R_{u,v}$ that is defined by the walk along the edge-segment of $(u,u')$ from $u$ to $c$, the edge segment of $(v,v')$ from $c$ to $v$ and the curve $[u,v]$ from $v$ to $u$. We call $[u,v]$ \emph{corner edge} with respect to~$(u,u')$ and $(v,v')$ if and only if $R_{u,v}$ has no vertices of $\Gamma(G)$ in its interior.

\begin{figure}[tb]
	\centering
    \begin{minipage}[b]{.18\textwidth}
        \centering
        \includegraphics[width=\textwidth,page=1]{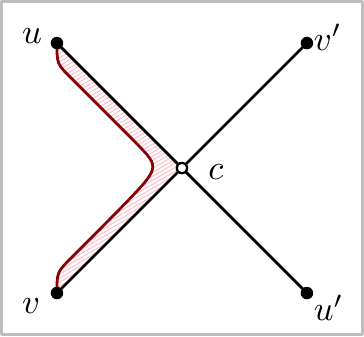}
        \subcaption{~}\label{fig:corner_pair}
    \end{minipage}
    \begin{minipage}[b]{.18\textwidth}
        \centering
        \includegraphics[width=\textwidth,page=2]{images/prop_corner}
        \subcaption{~}\label{fig:corner_pair_same}
    \end{minipage}
    \begin{minipage}[b]{.18\textwidth}
        \centering
        \includegraphics[width=\textwidth,page=1]{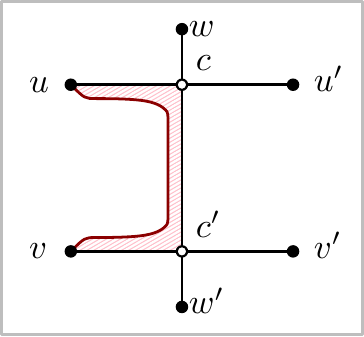}
        \subcaption{~}\label{fig:parallel_pair}
    \end{minipage}
	\begin{minipage}[b]{.18\textwidth}
        \centering
        \includegraphics[width=\textwidth,page=2]{images/prop_parallel}
        \subcaption{~}\label{fig:parallel_pair_same}
    \end{minipage}
	\begin{minipage}[b]{.18\textwidth}
        \centering
        \includegraphics[width=\textwidth,page=3]{images/prop_parallel}
        \subcaption{~}\label{fig:parallel_pair_homotopic}
    \end{minipage}
    \caption{%
    (a-b)~vertices $u$ and $v$ form a corner pair;
    (c-d)~vertices $u$ and $v$ form a \side pair;
    (e)~at least one of the two potential \ses exists.}
    \label{fig:crossing_confs}
\end{figure}

\begin{property}
In a PMCM-drawing $\Gamma(G)$ of an optimal $k$-planar graph $G$ any corner edge $[u,v]$ is a potential edge.
\label{prp:corner}
\end{property}
\begin{proof}
By the definition of potential edges, the property holds when $u \neq v$. Consider now the case where $u=v$. In this case $[u,v]$ is a self-loop; see Figure~\ref{fig:corner_pair_same}. If the property does not hold, then it follows that $[u,v]$ is a self-loop with no vertices either in its interior or in its exterior. However, this contradicts Lemma~\ref{lem:crossing_adjacent}, and the property holds.
\end{proof}

We say that vertices $u$ and $v$ form a \emph{\side pair} if and only if there exist edges $(u,u')$ and $(v,v')$ for some $u'$ and $v'$ such that they both cross a third edge $(w,w')$ in $\Gamma(G)$ and additionally $(u,u') \neq (v,v')$; see Figure~\ref{fig:parallel_pair} or \ref{fig:parallel_pair_same}. Let $c$ and $c'$ be the crossing points of $(u,u')$ and $(v,v')$ with $(w,w')$, respectively. Assume w.l.o.g.~that $c$ and $c'$ appear in this order along $(w,w')$ from vertex $w$ to vertex $w'$. Also assume that the edge-segment of $(u,u')$ between $u$ and $c$ is on the same side of edge $(w,w')$ as the edge-segment of $(v,v')$ between $v$ and $c'$; refer to Figure~\ref{fig:parallel_pair}. Then, any Jordan curve $[u,v]$ joining vertices $u$ and $v$ induces a region $R_{u,v}$ that is defined by the walk along the edge-segment of $(u,u')$ from $u$ to $c$, the edge segment of $(w,w')$ from $c$ to $c'$, the edge segment of $(v,v')$ from $c'$ to $v$ and the curve $[u,v]$ from $v$ to $u$. We call $[u,v]$ \emph{\se} w.r.t.~$(u,u')$ and $(v,v')$  if and only if $R_{u,v}$ has no vertices of $\Gamma(G)$ in its interior. Since by Properties~\ref{prp:2_planar_quasi} and \ref{prp:3_planar_quasi} edges $(u,u')$ and $(v,v')$ cannot cross with each other (as they both cross $(w,w')$), it follows that region $R_{u,v}$ is well-defined. Symmetrically we define region $R_{u',v'}$ and \se $[u',v']$ with respect to~$(u,u')$ and $(v,v')$.

\begin{property}
In a PMCM-drawing $\Gamma(G)$ of an optimal $k$-planar graph $G$ with $k \in \{2,3\}$ at least one of the \ses $[u,v]$, $[u',v']$ is a \pe.
\label{prp:parallel}
\end{property}
\begin{proof}
Before giving the proof, note that since edges $(u,u')$, $(v,v')$ and $(w,w')$ do not mutually cross, curves $[u,v]$ and $[u',v']$ cannot cross themselves. Now, for a proof by contradiction, assume that neither $[u,v]$ nor $[u',v']$ are potential edges. This implies that $u=v$, $u'=v'$ and both $[u,v]$ and $[u',v']$ are self-loops that have no vertices in their interiors or their exteriors. Figure~\ref{fig:parallel_pair_homotopic} illustrates the case where both $[u,v]$ and $[u',v']$ are self-loops with no vertices in their interiors; the other cases are similar. It is not hard to see that $(u,u')$ and $(v,v')$ are homotopic \ses; a contradiction.
\end{proof}

\noindent We say that $(u,u')$ and $(v,v')$ are \emph{\sa} if and only if both \ses $[u,v]$ and $[u',v']$ are \pes.

% ============================================================================
\section{Characterization of optimal 2-planar graphs}
\label{sec:2planar}
% ============================================================================

By using the properties we proved in Section~\ref{sec:properties}, in this section we examine some more structural properties of optimal $2$-planar graphs in order to derive their characterization (see Theorem~\ref{thm:2-characterization}).

\begin{lemma}
Let $\Gamma(G)$ be a PMCM-drawing of an optimal $2$-planar graph $G$. Any edge that is crossed twice in  $\Gamma(G)$ is a chord of a true-planar $5$-cycle in $\Gamma(G)$.
\label{lem:2_planar_small_faces}
\end{lemma}
\begin{proof}
Let $(u,v)$ be an edge of $G$ that is crossed twice in $\Gamma(G)$ by edges $(u',v')$ and $(u'',v'')$ at points $c$ and $c'$, respectively. Note that, by Property~\ref{prp:2_planar_cross_twice} edges  $(u',v')$ and $(u'',v'')$ are not identical. We assume w.l.o.g.~that $c$ and $c'$ appear in this order along $(u,v)$ from vertex $u$ to vertex $v$. We also assume that the edge-segment of $(u',v')$ between $u'$ and $c$ is on the same side of edge $(u,v)$ as the edge-segment of $(u'',v'')$ between $u''$ and $c'$; refer to Figure~\ref{fig:2_general}. By Property~\ref{prp:corner} corner edges $[u,u']$, $[u,v']$, $[v,u'']$ and $[v,v'']$ are \pes. By Property~\ref{prp:parallel} at least one of \ses $[u',u'']$ and $[v',v'']$ is a \pe. Assume w.l.o.g.~that $[v',v'']$ is a \pe.

\begin{figure}[t]
	\centering
    \begin{minipage}[b]{.18\textwidth}
        \centering
        \includegraphics[width=\textwidth,page=1]{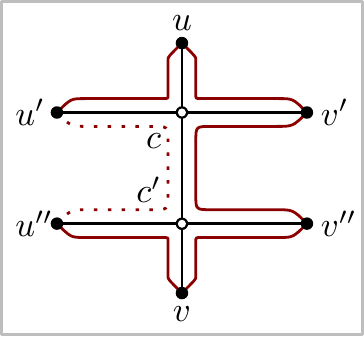}
        \subcaption{~}\label{fig:2_general}
    \end{minipage}
    \begin{minipage}[b]{.18\textwidth}
        \centering
        \includegraphics[width=\textwidth,page=2]{images/2planar_2_crossing}
        \subcaption{~}\label{fig:2_6gon}
    \end{minipage}
	\begin{minipage}[b]{.18\textwidth}
        \centering
        \includegraphics[width=\textwidth,page=3]{images/2planar_2_crossing}
        \subcaption{~}\label{fig:2_5gon}
    \end{minipage}
		\begin{minipage}[b]{.18\textwidth}
        \centering
        \includegraphics[width=\textwidth,page=4]{images/2planar_2_crossing}
        \subcaption{~}\label{fig:2_5gon_extend}
    \end{minipage}
			\begin{minipage}[b]{.18\textwidth}
        \centering
        \includegraphics[width=\textwidth,page=5]{images/2planar_2_crossing}
        \subcaption{~}\label{fig:2_5gon_final}
    \end{minipage}
    \caption{%
    Different configurations used in Lemma~\ref{lem:2_planar_small_faces}.}
    \label{fig:2_planar_potential_parallel}
\end{figure}

First consider the case that $[u',u'']$ is also a \pe; see Figure~\ref{fig:2_6gon}. In this case, vertices $u$, $v'$, $v''$, $v$, $u''$ and $u'$ define a \pec $\mathcal{C}$ on six vertices (shaded in gray in Figure~\ref{fig:2_6gon}). Edges $(u,v)$, $(u',v')$ and $(u'',v'')$ are drawn in the interior of $\mathcal{C}$, and there exist at most two other edges that cross $(u',v')$ or $(u'',v'')$. In total there exist at most five edges that have an edge-segment within $\mathcal{C}$. However, in the interior of $\mathcal{C}$ one can draw six chords as in Figure~\ref{fig:2_planar_6gon} without deviating $2$-planarity. By Lemma~\ref{lem:exchange}.(\ref{prp:nonoptim}) for $\kappa+\lambda\leq5$ and $\mu=6$, it follows that $G$ is not optimal; a contradiction.

To complete the proof, it remains to consider the cases where $[u',u'']$ is not a \pe; see Figure~\ref{fig:2_5gon}. In this case, $[u',u'']$ is a homotopic self-loop (hence, the red-shaded region of Figure~\ref{fig:2_5gon} contains no vertices in its interior). Vertices $u$, $v'$, $v''$, $v$ and $u'$ define a \pec $\mathcal{C}$ on five vertices (shaded in gray in Figure~\ref{fig:2_5gon}). However, in the interior of $\mathcal{C}$ one can draw five chords as in Figure~\ref{fig:2_planar_5gon} without deviating $2$-planarity. By Lemma~\ref{lem:exchange}.(\ref{prp:boundary}), for $\kappa+\lambda\leq 5$ and $\mu=5$, it follows that all boundary edges of $\mathcal{C}$ exist in $\Gamma(G)$. Furthermore, $\kappa+\lambda=5$ must hold, which implies that there exist two edges (other than $(u,v)$), say $e$ and $e'$, that cross $(u',v')$ and $(u'',v'')$ respectively.

If $\mathcal{C}$ is a true-planar $5$-cycle in $\Gamma(G)$ the lemma holds. If it is not, then at least one of edges $e$ or $e'$ crosses a boundary edge of $\mathcal{C}$. Suppose w.l.o.g.~that edge $e$ crosses $(v',v'')$ of $\mathcal{C}$ at point $p$ and let $w$ and $w'$ be the endpoints of $e$ (other cases are similar). Observe that $e$ already has two crossings in $\Gamma(G)$. By $2$-planarity, either the edge-segment of $(w,w')$ between $w$ and $p$ or the one between $w'$ and $p$ is drawn completely in the exterior of $\mathcal{C}$. Suppose w.l.o.g.~that this edge-segment is the one between $w$ and $p$. Then vertices $v'$, $w$ and $v''$ define a \pec $\mathcal{C}'$ on three vertices; see Figure~\ref{fig:2_5gon_extend}. We proceed as follows: We remove edges $(u,v)$, $(u',v')$, $(u'',v'')$, $e$ and $e'$ and replace them with five chords drawn in the interior of $\mathcal{C}$ (as in Figure~\ref{fig:2_5gon_final}). The derived graph $G'$ has the same number of edges as $G$. However, $\mathcal{C'}$ becomes a true-planar $3$-cycle in $G'$, contradicting Property~\ref{prp:2planar_triangle}.
\end{proof}

By Lemma~\ref{lem:2_planar_small_faces}, any edge of $G$ that is crossed twice in $\Gamma(G)$ is a chord of a true-planar $5$-cycle. So, it remains to consider edges of $G$ that have only one crossing in $\Gamma(G)$. In fact, the following lemma states that there are no such edges in $\Gamma(G)$.

\begin{lemma}
Let $\Gamma(G)$ be a PMCM-drawing of an optimal $2$-planar graph $G$. Then, every edge of $\Gamma(G)$ is either true-planar or has exactly two crossings.
\label{lem:2_planar_one_crossing}
\end{lemma}
\begin{proof}
As shown in the proof of Lemma~\ref{lem:2_planar_small_faces}, for any edge $e$ of $G$ that is crossed twice in $\Gamma(G)$, both edges that cross $e$ also have two crossings in $\Gamma(G)$. So, the crossing component $\mathcal{X}(e)$ consists exclusively of edges with two pairwise crossings. This implies that if edges $(u,v)$ and $(u',v')$ cross in $\Gamma(G)$ and $(u,v)$ has only one crossing, then the same holds for $(u',v')$; see Figure~\ref{fig:2_planar_one_crossing_before}. Vertices $u$, $v'$, $v$ and $u'$ define a \pec $\mathcal{C}$ on four vertices (gray-shaded in Figure~\ref{fig:2_planar_one_crossing_before}). Since edges $(u,v)$ and $(u',v')$ have only one crossing each, the boundary of $\mathcal{C}$ exists in $\Gamma(G)$ and are true-planar edges. We proceed by removing edge $(u',v')$. Now $\mathcal{C}$ is split into two true-planar $3$-cycles; see Figure~\ref{fig:2_planar_one_crossing_after}. In both of them, we plug the $2$-planar pattern of Figure~\ref{fig:2_planar_triangle}. In total, we removed one edge and added two vertices and a total of $12$ edges, without creating any homotopic parallel edges or self-loops. Hence, if $G$ has $n$ vertices and $m$ edges, the derived graph $G'$ is $2$-planar and has $n'=n+2$ vertices and $m'=m+11$ edges. Hence $m'=5n'-9$, i.e. $G'$ has more edges than allowed; a contradiction.
\end{proof}

\begin{figure}[t]
	\centering
    \begin{minipage}[b]{.18\textwidth}
        \centering
        \includegraphics[width=\textwidth,page=1]{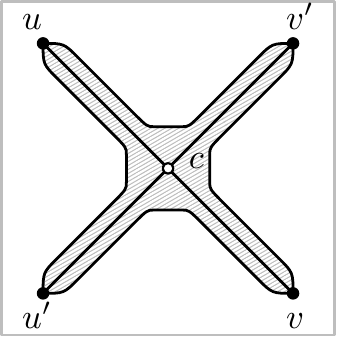}
        \subcaption{~}\label{fig:2_planar_one_crossing_before}
    \end{minipage}
    \begin{minipage}[b]{.18\textwidth}
        \centering
        \includegraphics[width=\textwidth,page=2]{images/2planar_one_crossing}
        \subcaption{~}\label{fig:2_planar_one_crossing_after}
    \end{minipage}
    \caption{%
    Different configurations used in~Lemma~\ref{lem:2_planar_one_crossing}.}
    \label{fig:2_planar_one_crossing}
\end{figure}

\begin{lemma}
The true-planar skeleton $\Pi(G)$ of a PMCM-drawing $\Gamma(G)$ of an optimal $2$-planar graph is connected.
\label{prp:2_planar_skeleton_connected}
\end{lemma}
\begin{proof}
Assume to the contrary that $\Pi(G)$ is not connected and let $H$ be a connected component of $\Pi(G)$.
By  Property~\ref{prp:connected} either there exists an edge $(u,v)$ with $u \in H$ and $v \in G \setminus H$, or two crossing edges $e_1 \in H$ and $e_2 \in G \setminus H$. In the first case, $(u,v)$ is not a true-planar edge. By Lemma~\ref{lem:2_planar_small_faces}, there exists a true-planar $5$-cycle with chord $(u,v)$ connecting $u$ to $v$ in $\Pi(G)$; a contradiction.
In the second case, edges $e_1$ and $e_2$ belong to the same crossing component and by Lemma~\ref{lem:2_planar_small_faces}, there exists a true-planar $5$-cycle with $e_1$ and $e_2$ as chords, therefore connecting their endpoints in $\Pi(G)$; a contradiction.
\end{proof}

\begin{lemma}
The true-planar skeleton $\Pi(G)$ of a PMCM-drawing $\Gamma(G)$ of an optimal $2$-planar graph $G$ contains only faces of length $5$, each of which contains $5$ crossing edges in $\Gamma(G)$.
\label{lem:2_planar_faces}
\end{lemma}
\begin{proof}
Since $\Pi(G)$ is connected (by Lemma~\ref{prp:2_planar_skeleton_connected}), all faces of $\Pi(G)$ are also connected. By Lemmas~\ref{lem:2_planar_small_faces} and \ref{lem:2_planar_one_crossing}, all crossing edges are chords of true-planar $5$-cycles. We claim that $\Pi(G)$ has no chordless faces. First, $\Pi(G)$ cannot contain a chordless face of size $\geq 4$, as otherwise we could draw in its interior a chord, contradicting the optimality of $G$. Property~\ref{prp:2planar_triangle} ensures that $\Pi(G)$ contains no faces of size $3$. Finally, faces of size $1$ or $2$ correspond to homotopic self-loops and parallel~edges.
\end{proof}

\noindent We are now ready to state the main theorem of this section.

\begin{theorem}
A graph $G$ is optimal $2$-planar if and only if $G$ admits a drawing $\Gamma(G)$ without homotopic parallel edges and self-loops, such that the true-planar skeleton $\Pi(G)$ of $\Gamma(G)$ spans all vertices of $G$, it contains only faces of length $5$ (that are not necessarily simple), and each face of $\Pi(G)$ has $5$ crossing edges in its interior in $\Gamma(G)$.
\label{thm:2-characterization}
\end{theorem}
\begin{proof}
For the forward direction, consider an optimal $2$-planar graph $G$. By Lemma~\ref{lem:2_planar_faces}, the true-planar skeleton $\Pi(G)$ of its $2$-planar PMCM-drawing $\Gamma(G)$ contains only faces of length $5$ and each face of $\Pi(G)$ has $5$ crossing edges in its interior in $\Gamma(G)$. Since the endpoints of two crossing edges are within a true-planar $5$-cycle (by Lemmas~\ref{lem:2_planar_small_faces} and \ref{lem:2_planar_one_crossing}) and since $\Pi(G)$ is connected (by Lemma~\ref{prp:2_planar_skeleton_connected}), $\Pi(G)$ spans all vertices of $G$. This completes the proof of this direction.

For the reverse direction, denote by $n$, $m$ and $f$ the number of vertices, edges and faces of $\Pi(G)$. Since $\Pi(G)$ spans all vertices of $G$, it suffices to prove that $G$ has exactly $5n-10$ edges. The fact that $\Pi(G)$ contains only faces of length $5$ implies that $5f=2m$. By Euler's formula for planar graphs, $m=5(n-2)/3$ and $f=2(n-2)/3$ follows. Since each face of $\Pi(G)$ contains exactly $5$ crossing edges, the total number of edges of $G$ equals $m+5f=5n-10$.
\end{proof}

% ============================================================================
\section{Characterization of optimal 3-planar graphs}
\label{sec:3planar}
% ============================================================================

In this section we explore several structural properties of optimal $3$-planar graphs to derive their characterizations (see Theorem~\ref{thm:3-characterization}).

\begin{lemma}
Let $\Gamma(G)$ be a PMCM-drawing of an optimal $3$-planar graph $G$, and suppose that there exists a \pec $\mathcal{C}$ of $6$ vertices in $\Gamma(G)$, such that the potential boundary edges of $\mathcal{C}$ exist in $\Gamma(G)$. Let $E_{\mathcal{C}}$ be the set of edge-segments within $\mathcal{C}$. If the conditions C.\ref{cnd:1} and C.\ref{cnd:2} hold, then $\mathcal{C}$ is an empty true-planar $6$-cycle in $\Gamma(G)$ and all edges with edge-segments in $E_{\mathcal{C}}$ are drawn as chords in its interior.
\begin{enumerate}[C.1:]
\item \label{cnd:1} $|E_{\mathcal{C}}|\leq 8$, and,
\item \label{cnd:2} every edge-segment of $E_{\mathcal{C}}$ has at least one crossing in the interior of $\mathcal{C}$.
\end{enumerate}
\label{lem:size9}
\end{lemma}
\begin{proof}
We start with the following observation: If $e$ is an edge of $G$, then due to $3$-planarity at most one edge-segment of $e$ belongs to $E_{\mathcal{C}}$. More precisely, if $E_{\mathcal{C}}$ contains at least two edge-segments of $e$, then we claim that $e$ has at least four crossings. By Condition C.\ref{cnd:2} each of the two edge-segments of $e$ contributes one crossing to $e$. Since $\mathcal{C}$ is empty and contains two edge-segments of $e$, it follows that $e$ exists and enters $\mathcal{C}$. Hence, $e$ has two more crossings, summing up to a total of at least four crossings.

Let $v_1,\dots,v_6$ be the vertices of $\mathcal{C}$. If all edges with edge-segments in $E_{\mathcal{C}}$ completely lie in $\mathcal{C}$,  then $\mathcal{C}$ is a true-planar $6$-cycle and the lemma trivially holds. Otherwise, there is at least one edge $e$ with an edge-segment in $E_{\mathcal{C}}$, that crosses a boundary edge of $\mathcal{C}$. W.l.o.g.~we can assume that $e$ crosses $(v_1,v_6)$ of $\mathcal{C}$ at point $c$ (refer to Figure~\ref{fig:cs1}). If $w$ and $w'$ are the two endpoints of $e$, then by the observation we made at the beginning of the proof it follows that either the edge-segment of $(w,w')$ between $w$ and $c$ or the one between $c$ and $w'$ is drawn completely in the exterior of $\mathcal{C}$ (as otherwise $e$ would have at least two edge-segments in $E_{\mathcal{C}}$). W.l.o.g.~assume that this is the edge-segment between $w$ and $c$. Then, corner edges $[v_1,w]$ and $[w,v_6]$ are \pes (by Property~\ref{prp:corner}).

\begin{figure}[t!]
	\centering
    \begin{minipage}[b]{.16\textwidth}
        \centering
        \includegraphics[width=\textwidth,page=1]{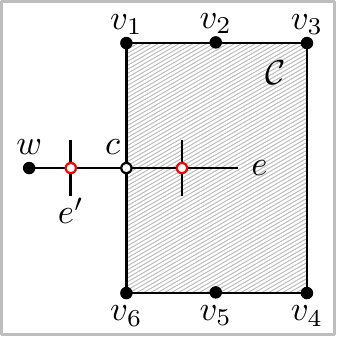}
        \subcaption{~}\label{fig:cs1}
    \end{minipage}
    \begin{minipage}[b]{.16\textwidth}
        \centering
        \includegraphics[width=\textwidth,page=2]{images/3planar_6gon_bound}
        \subcaption{~}\label{fig:cs2}
    \end{minipage}
	\begin{minipage}[b]{.16\textwidth}
        \centering
        \includegraphics[width=\textwidth,page=5]{images/3planar_6gon_bound}
        \subcaption{~}\label{fig:7_stick}
    \end{minipage}
	\begin{minipage}[b]{.16\textwidth}
        \centering
        \includegraphics[width=\textwidth,page=3]{images/3planar_6gon_bound}
        \subcaption{~}\label{fig:cp}
    \end{minipage}
    \begin{minipage}[b]{.16\textwidth}
        \centering
        \includegraphics[width=\textwidth,page=4]{images/3planar_6gon_bound}
        \subcaption{~}\label{fig:cs_final}
    \end{minipage}
    \begin{minipage}[b]{.16\textwidth}
        \centering
        \includegraphics[width=\textwidth,page=1]{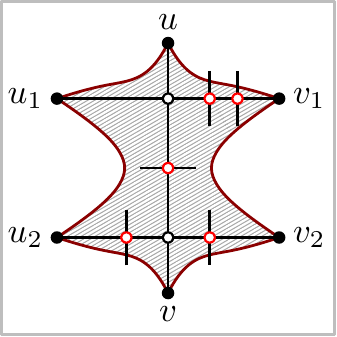}
        \subcaption{~}\label{fig:independent}
    \end{minipage}
    \caption{%
    Different configurations used in
    (a)-(d)~Lemma~\ref{lem:size9},
    (e)~Lemma~\ref{lem:3_planar_independent}.}
    \label{fig:replacements_2}
\end{figure}

Recall that $e$ has one crossing in the interior of $\mathcal{C}$ (by Condition C.\ref{cnd:2} of the lemma) and one more crossing with edge $(v_1,v_6)$. By $3$-planarity, it follows that edge $e$ may have at most one more crossing, say with edge $e'$. Note that $e'$ may or may not have an edge-segment in $E_{\mathcal{C}}$. Vertices $w$, $v_1$, $\dots$, $v_6$ define a \pec $\mathcal{C}'$ on $7$ vertices (see Figure~\ref{fig:cs2}). The set $E_{\mathcal{C}'}$ of edge-segments within $\mathcal{C}'$ contains all edge-segments of $E_{\mathcal{C}}$ (that is, $E_{\mathcal{C}} \subseteq E_{\mathcal{C}'}$) plus at most two additional edge-segments: the one defined by edge $(v_1,v_6)$, and possibly an edge-segment of $e'$. Hence $|E_{\mathcal{C}'}| \leq 10$. In the following we make some observations in the form of claims.

\begin{claim}
All edges with an edge-segment in $E_{\mathcal{C}'}$ have at least one crossing in the interior of $\mathcal{C}'$.
\label{nclm:1}
\end{claim}
\begin{proof}
The claim clearly holds for all edge-segments of $E_{\mathcal{C}}$ (recall that $E_{\mathcal{C}}\subset E_{\mathcal{C}'}$). Since $(v_1,v_{6})$ and $e'$ (if it exists) both cross $e$ in the interior of $\mathcal{C}'$, the remaining edge-segments within $\mathcal{C}'$ (i.e., the ones defined by edges $(v_1,v_{6})$ and $e'$) have at least one crossing in the interior of $\mathcal{C}'$.
\end{proof}

\begin{claim}
At least one edge with an edge-segment in $E_{\mathcal{C}'}$ crosses one edge of $\mathcal{C}'$.
\label{nclm:2}
\end{claim}
\begin{proof}
If all edges with an edge-segment in $E_{\mathcal{C}'}$ do not cross $\mathcal{C}'$, then all edges with an edge-segment in $E_{\mathcal{C}'}$ can be drawn completely in the interior of $\mathcal{C}'$. Hence, all \pes of $\mathcal{C}'$ can be added in $\Gamma(G)$ (if they are not present already). Then, $\mathcal{C}'$ is a true-planar $7$-cycle contradicting Property~\ref{prp:3planar_odd_cycle}.
\end{proof}

\noindent By Claim~\ref{nclm:2}, there is an edge $g$ that crosses a boundary edge, say $[w,v_1]$, of $\mathcal{C}'$ at point $c'$; see Figure~\ref{fig:7_stick}.

\begin{claim}
All boundary edges of $\mathcal{C}'$ exist in $\Gamma(G)$ and $g$ has one crossing in the interior of $\mathcal{C}'$.
\label{nclm:3}
\end{claim}
\begin{proof}
To prove this claim, we remove all edges with an edge-segment in $E_{\mathcal{C}'}$ (recall that $|E_{\mathcal{C}'}| \leq 10$) and replace them with the $10$ edges of the $3$-planar crossing pattern of Figure~\ref{fig:cp}, i.e., we redraw the segment of $g$ in the interior of $\mathcal{C}'$ so that: %
\begin{inparaenum}[(i)]
\item $g$ emanates from vertex $v_6$ of $\mathcal{C}'$,
\item $g$ crosses only \pe $[w,v_1]$ at point $c'$, and
\item $g$ has no other crossings in the interior of  $\mathcal{C}'$.
\end{inparaenum}
This allows us to add all boundary edges of  $\mathcal{C}'$ in $\Gamma(G)$ (if they are not present). Hence, $3$-planarity is preserved and the derived graph has at least as many edges as $G$. Since $G$ is optimal, it follows that all boundary edges of $\mathcal{C}'$ must exist in $\Gamma(G)$, which completes the proof of the claim.
\end{proof}

We follow an analogous approach to the one we used for expanding $\mathcal{C}$ (that has $6$ vertices) to $\mathcal{C}'$ (that has $7$ vertices). We can find an endpoint of $g$, say $z$, such that $w$, $z$, $v_1$, $v_2$, $\dots$, $v_6$ define a \pec $\mathcal{C}''$ on $8$ vertices. Furthermore, the set $E_{\mathcal{C}''}$ of edge-segments within $\mathcal{C}''$ has at most $12$ elements (at most two more than $E_{\mathcal{C}'}$). We proceed by removing all edges with an edge-segment in $E_{\mathcal{C}''}$ and split $\mathcal{C}''$ into two true-planar cycles of length $6$ and $4$, by adding  true-planar chord $(v_1,v_6)$; see Figure~\ref{fig:cs_final}. In the interior of the $6$-cycle, we add $8$ crossing edges as in Figure~\ref{fig:3_planar_6gon}. In the interior of the $4$-cycle, we add a vertex $x$ with a true planar edge $(v_1,x)$. Vertices $v_1$, $x$, $v_1$, $v_6$, $w$ and $z$ define a new \pec on $6$ vertices, allowing us to add $8$ more crossing edges. In total, we removed at most $12$ edges, added a vertex and $18$ edges. If $n$ and $m$ are the number of vertices and edges of $G$, then the derived graph $G'$ has $n'=n+1$ vertices and $m'\geq m+6$ edges. The last equation gives $m'\geq 5.5n'-10.5$, i.e. $G'$ has more edges than allowed; a contradiction.
\end{proof}

Let $(u,v)$ be an edge of $G$ that is crossed by two edges $(u_1,v_1)$ and $(u_2,v_2)$ in $\Gamma(G)$ at points $c_1$ and $c_2$. By Property~\ref{prp:3_planar_cross_twice} edges $(u_1,v_1)$ and $(u_2,v_2)$ are not identical. We assume w.l.o.g.~that $c_1$ and $c_2$ appear in this order along $(u,v)$ from $u$ to $v$. We also assume that the edge-segment of $(u_1,v_1)$ between $u_1$ and $c$ is on the same side of edge $(u,v)$ as the edge-segment of $(u_2,v_2)$ between $u_2$ and $c_2$; refer to Figure~\ref{fig:independent}. Vertices $u_1$, $u_2$ and $v_1$, $v_2$ define two \side pairs. By Property~\ref{prp:parallel}, at least one of \ses $[u_1,u_2]$ and $[v_1,v_2]$ is a \pe of $\Gamma(G)$. Recall that if both \ses $[u_1,u_2]$ and $[v_1,v_2]$ are \pes of $\Gamma(G)$, then edges $(u_1,v_1)$ and $(u_2,v_2)$ are called \sa.

\begin{lemma}
Let $\Gamma(G)$ be a PMCM-drawing of an optimal $3$-planar graph~$G$. If $(u,v)$ is crossed by \sa edges $(u_1,v_1)$ and $(u_2,v_2)$ in $\Gamma(G)$, then it is a chord of an empty true-planar~$6$-cycle.
\label{lem:3_planar_independent}
\end{lemma}
\begin{proof}
Refer to Figure~\ref{fig:independent}. Since $(u_1,v_1)$ and $(u_2,v_2)$ are \sa, \ses $[u_1,u_2]$ and $[v_1,v_2]$ are \pes. By Property~\ref{prp:corner}, corner edges $[u, u_1]$, $[u,v_1]$, $[u,u_2]$ and $[v,v_2]$ are \pes. Hence, vertices $u$, $v_1$, $v_2$, $v$, $u_2$ and $u_1$ define a \pec $\mathcal{C}$ on six vertices (gray-shaded in Figure~\ref{fig:independent}). Edges $(u,v)$, $(u_1,v_1)$ and $(u_2,v_2)$ are drawn completely in the interior of $\mathcal{C}$ and there exist at most five other edges either drawn in the interior of $\mathcal{C}$ or crossing its boundary:  at most one that crosses $(u,v)$, and at most four others that cross $(u_1,v_1)$ and $(u_2,v_2)$. Since we can draw eight chords in the interior of $\mathcal{C}$ as in Figure~\ref{fig:3_planar_6gon}, by Lemma~\ref{lem:exchange}.(\ref{prp:boundary}), for $\kappa+\lambda\leq 8$ and $\mu=8$, all boundary edges of $\mathcal{C}$ exist in $\Gamma(G)$. Furthermore $\kappa+\lambda=8$ must hold. Note that the set $E_{\mathcal{C}}$ of edge-segments within $\mathcal{C}$ contains only edge-segments of these $\kappa+\lambda$ edges. Also, these $8$ edges have exactly one edge-segment within $\mathcal{C}$ that is crossed in the interior of $\mathcal{C}$. Hence, conditions C.\ref{cnd:1} and C.\ref{cnd:2} of Lemma~\ref{lem:size9} are satisfied and there exists an empty true-planar $6$-cycle that has $(u,v)$ as chord.
\end{proof}

\begin{lemma}
Let $\Gamma(G)$ be a PMCM-drawing of an optimal $3$-planar graph~$G$. If $e$ is crossed by two \sa edges in $\Gamma(G)$, then all edges of $\mathcal{X}(e)$ are chords of an empty true-planar $6$-cycle.
\label{lem:3_planar_xing_comp}
\end{lemma}
\begin{proof}
The lemma follows by the observation that since $e$ is a chord of an empty true-planar $6$-cycle (by Lemma~\ref{lem:3_planar_independent}), all edges of $\mathcal{X}(e)$ are also chords of this $6$-cycle.
\end{proof}

\begin{lemma}
Let $\Gamma(G)$ be a PMCM-drawing of an optimal $3$-planar graph $G$. Any edge that is crossed three times in $\Gamma(G)$ is a chord of an empty true-planar $6$-cycle in $\Gamma(G)$.
\label{lem:3_planar_small_faces}
\end{lemma}
\begin{proof}
Our proof is based on a case analysis and in order to lighten the presentation we will use intermediate observations in the form of claims. Let $(u,v)$ be an edge of $G$ that crosses edges $(u_i,v_i)$ in $\Gamma(G)$, for $i=1,2,3$. Let also $c_1$, $c_2$ and $c_3$ be the corresponding crossing points as they appear along $(u,v)$ from vertex $u$ to vertex $v$; see Figure~\ref{fig:3_planar_three_crossing}. We assume w.l.o.g.~that the edge-segment of $(u_i,v_i)$ between $u_i$ and $c_i$ is on the same side of edge $(u,v)$ as the edge-segment of $(u_j,v_j)$ between $u_j$ and $c_j$, for $1\leq i<j\leq 3$. Consider the crossing component $\mathcal{X}((u,v))$. We distinguish two cases depending on whether there exists an edge in $\mathcal{X}((u,v))$ that crosses two \sa edges or not. Assume that there is an edge of $\mathcal{X}((u,v))$ that crosses two \sa edges. Then, by Lemma~\ref{lem:3_planar_xing_comp} all edges of $\mathcal{X}((u,v))$, including $(u,v)$, are chords of an empty true-planar $6$-cycle and the lemma follows. Assume now that there exists no edge in $\mathcal{X}((u,v))$ that crosses two \sa edges. Hence, for edge $(u,v)$, that crosses edges $(u_1,v_1)$, $(u_2,v_2)$ and $(u_3,v_3)$, we have that any two edges $(u_i,v_i)$ and $(u_j,v_j)$ with $1\leq i<j\leq 3$, are not \sa. Observe that by definition, exactly one of \ses $[u_i,u_j]$ or $[v_i,v_j]$ is not a \pe. In the following claim, we refine this observation.

\begin{claim}
Either \ses $[u_1,u_2]$, $[u_1,u_3]$ and $[u_2,u_3]$ are not \pes or \ses $[v_1,v_2]$, $[v_1,v_3]$ and $[v_2,v_3]$ are not \pes.
\label{clm:fan_planar}
\end{claim}
\begin{proof}
Consider \ses $[u_1,u_2]$ and $[u_2,u_3]$ and assume that both are not \pes. It follows that $[u_1,u_2]$ and $[u_2,u_3]$ are both homotopic self-loops. Hence, $u_1=u_2=u_3$. We will prove that \se $[u_1,u_3]$ is not a \pe either. Let $R_{1,2}$ be the region defined by the edge-segment of $(u_1,v_1)$ between $u_1$ and crossing $c_1$, the edge-segment of $(u,v)$ between crossings $c_1$ and $c_2$ and the edge-segment of $(u_2,v_2)$ between $c_2$ and $u_2$ (recall that $u_1=u_2$; see Figure~\ref{fig:3_planar_small_faces_example}). Similarly, we define regions $R_{2,3}$ and $R_{1,3}$. Observe that $R_{1,2}\cup R_{2,3} = R_{1,3}$. Since $[u_1,u_2]$ and $[u_2,u_3]$ are homotopic self-loops, $R_{1,2}$ and $R_{2,3}$ do not contain any vertex in their interiors. Hence, $R_{1,3}$ does not contain any vertex in its interior either. More in general, since $R_{1,2} \cup R_{2,3} = R_{1,3}$ we can prove that whenever any two of $[u_1,u_2]$, $[u_2,u_3]$ and $[u_1,u_3]$ are not \pes, then the third one is not a \pe either.  Similarly, we can prove that whenever any two of $[v_1,v_2]$, $[v_2,v_3]$ and $[v_1,v_3]$ are not \pes, then the third one is not a \pe either.

\begin{figure}[t]
	\centering
    \begin{minipage}[b]{.18\textwidth}
        \centering
		\includegraphics[width=\textwidth,page=1]{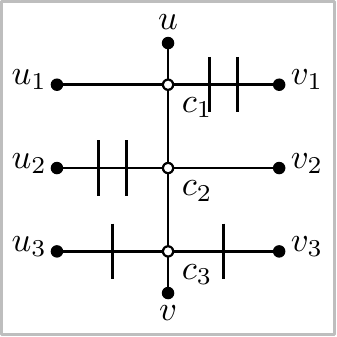}
        \subcaption{~}\label{fig:3_planar_three_crossing}
    \end{minipage}
    \begin{minipage}[b]{.18\textwidth}
        \centering
        \includegraphics[width=\textwidth,page=1]{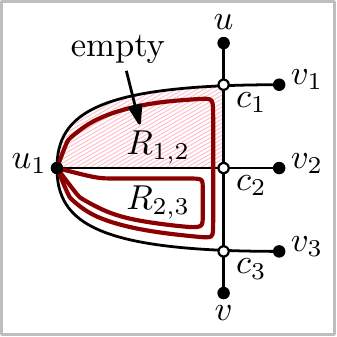}
        \subcaption{~}\label{fig:3_planar_small_faces_example}
    \end{minipage}
    \begin{minipage}[b]{.18\textwidth}
        \centering
        \includegraphics[width=\textwidth,page=2]{images/3planar_fan_crossing}
        \subcaption{~}\label{fig:3fan}
    \end{minipage}
    \begin{minipage}[b]{.18\textwidth}
        \centering
        \includegraphics[width=\textwidth,page=3]{images/3planar_fan_crossing}
        \subcaption{~}\label{fig:3fan_middle_a}
    \end{minipage}
	\begin{minipage}[b]{.18\textwidth}
        \centering
        \includegraphics[width=\textwidth,page=4]{images/3planar_fan_crossing}
        \subcaption{~}\label{fig:3fan_middle_b}
    \end{minipage}
    \caption{%
    Different configurations used in Lemma~\ref{lem:3_planar_small_faces}.}
\end{figure}

Finally, we show that at least two of $[u_1,u_2]$, $[u_1,u_3]$ and $[u_2,u_3]$  or at least two of $[v_1,v_2]$, $[v_1,v_3]$ and $[v_2,v_3]$ are not \pes, which, by our previous arguments, implies that the third \se is not a \pe either. If for example $[u_1,u_2]$ and $[u_1,u_3]$ are \pes, then neither $[v_1,v_2]$ nor $[v_1,v_3]$ is a \pe, as otherwise, either $(u_1,v_1)$ and $(u_2,v_2)$ are \sa or $(u_1,v_1)$ and $(u_3,v_3)$ are \sa, contradicting our previous observation.
\end{proof}

By Claim~\ref{clm:fan_planar} we can assume w.l.o.g.~that \ses $[u_1,u_2]$, $[u_1,u_3]$ and $[u_2,u_3]$ are not \pes in $\Gamma(G)$. This implies that regions $R_{1,2}$, $R_{2,3}$ and $R_{1,3}$ do not contain any vertex in their interiors (and also $u_1=u_2=u_3$). Hence, each edge of $\mathcal{X}(e)$ which is crossed by three edges in $\Gamma(G)$ complies with the crossing pattern of Figure~\ref{fig:3fan}, where the red-shaded region has no vertices in its interior. Now, vertices $u$, $v_1$, $v_2$, $v$, $u_2$ and $u_1$ define a \pec $\mathcal{C}$ on six vertices. Our goal is to use Lemma~\ref{lem:size9}, whose precondition C.\ref{cnd:1} requires at most $8$ edge-segments within $\mathcal{C}$. Note that since $(u,v)$ has three crossings and since each of $(u_1,v_1)$, $(u_2,v_2)$ and $(u_3,v_3)$ has one crossing, there may exist at most $10$ with at least one edge-segment within $\mathcal{C}$; see also Figure~\ref{fig:3_planar_three_crossing}. In the following, we prove that this is not the case.

\begin{claim}
Any edge crossing $(u_2,v_2)$ in the interior of $\mathcal{C}$ must also cross $(u_1,v_1)$ or $(u_3,v_3)$.
\label{clm:shared_crossing}
\end{claim}
\begin{proof}
Suppose that edge $(u',v')$ crosses $(u_2,v_2)$ at point $c_2'$ in the interior of $\mathcal{C}$. Recall that $c_2$ denotes the crossing point between $(u_2,v_2)$ and $(u,v)$. Since $(u_2,v_2) \in \mathcal{X}((u,v))$, edge $(u_2,v_2)$ is not crossed by \sa edges. So, edges $(u',v')$ and $(u,v)$ are not \sa, and exactly one of \ses $[u,u']$ or $[v,v']$ is not a \pe. Assume w.l.o.g.~that \se $[u,u']$ is not a \pe; see Figure~\ref{fig:3fan_middle_a}. This implies that $u=u'$ and that the region $R_{u,u'}$ defined by the edge-segment of $(u,v)$ between $u$ and $c_2$, the edge-segment of $(u_2,v_2)$ between $c_2$ and $c_2'$ and the edge-segment of $(u',v')$ between $c_2'$ and $u'$ has no vertices in its interior (red-shaded in Figure~\ref{fig:3fan_middle_a}). Then, edge $(u',v')$ must cross $(u_1,v_1)$, as otherwise vertex $v_1$ would be in the interior of $R_{u,u'}$; see Figure~\ref{fig:3fan_middle_b}. This completes the proof of this claim.
\end{proof}

Recall that our goal is to use Lemma~\ref{lem:size9}. Claim~\ref{clm:shared_crossing} implies that there exist at most four other edges that cross edges $(u_1,v_1)$, $(u_2,v_2)$ or $(u_3,v_3)$, i.e. we have at most $8$ edges that are either drawn in the interior of $\mathcal{C}$ or cross its boundary. Since one can draw eight chords in the interior of $\mathcal{C}$ as in Figure~\ref{fig:3_planar_6gon}, by Lemma~\ref{lem:exchange}.(\ref{prp:boundary}), for $\kappa+\lambda\leq8$ and $\mu=8$, it follows that the boundary edges of $\mathcal{C}$ exist in $\Gamma(G)$. Furthermore $\kappa+\lambda=8$ must hold. Note that the set $E_{\mathcal{C}}$ of edge-segments within $\mathcal{C}$ contains only edge-segments of these $\kappa+\lambda$ edges. Also, these $8$ edges have exactly one edge-segment within $\mathcal{C}$ that is crossed in the interior of $\mathcal{C}$. Hence conditions C.\ref{cnd:1} and C.\ref{cnd:2} of Lemma~\ref{lem:size9} are satisfied, and therefore we conclude that $(u,v)$ is a chord of a true planar $6$-cycle.
\end{proof}

\noindent By Lemma~\ref{lem:3_planar_small_faces}, any edge of $G$ that is crossed three times in $\Gamma(G)$ is a chord of an empty true-planar $6$-cycle. In the following, we consider edges of $G$ that have two or fewer crossings in $\Gamma(G)$. Hence, their crossing components contain edges with at most two crossings. Our approach is slightly different than the one we followed in the proof of Lemma~\ref{lem:2_planar_small_faces} for the optimal $2$-planar graphs.

\begin{lemma}
Let $\Gamma(G)$ be a PMCM-drawing of an optimal $3$-planar graph $G$ and let $\mathcal{X}$ be a crossing component of $\Gamma(G)$. Then, there is at least one edge in $\mathcal{X}$ that has three crossings.
\label{lem:3_planar_small_faces_2}
\end{lemma}
\begin{proof}
Assume to the contrary that there exists a crossing component $\mathcal{X}$ where all edges have at most two crossings. We distinguish two cases depending on whether $\mathcal{X}$ contains an edge with two crossings or not. Assume first that $\mathcal{X}$ does not contain an edge with two crossings. Then, $|\mathcal{X}|=2$. W.l.o.g.~assume that $\mathcal{X}=\{e,e'\}$. The four endpoints of edges $e$ and $e'$ define a \pec $\mathcal{C}$ on $4$ vertices; see Figure~\ref{fig:3_planar_one_crossing_before}. Since $e$ and $e'$ have only one crossing each, the \pes of the boundary of $\mathcal{C}$ exist in $\Gamma(G)$ and are true-planar edges. Note that there are no other edges passing through the interior of $\mathcal{C}$. We proceed by removing edges $e$ and $e'$ and replace them with the $3$-planar pattern of Figure~\ref{fig:3_planar_one_crossing_after}. In particular we add a vertex $x$ in the interior of $\mathcal{C}$ and true-planar edge $(v',x)$. Vertices $u$, $v'$, $x$, $v'$, $v$ and $u'$ define a \pec on six vertices, and we can add $8$ crossing edges in its interior as in Figure~\ref{fig:3_planar_6gon}. If $G$ has $n$ vertices and $m$ edges, the derived graph $G'$ has $n'=n+1$ vertices and $m'=m-2+8$ edges Then, $G'$ is $3$-planar and has $m'=5.5n'-10.5$ edges, that is, $G'$ has more edges than allowed by $3$-planarity; a contradiction.

To complete the proof, assume that there exists an edge $(u,v)\in \mathcal{X}$ which has two crossings, say with $(u_1,v_1)$ and $(u_2,v_2)$. By Lemma~\ref{lem:3_planar_independent}, $(u_1,v_1)$ and $(u_2,v_2)$ are not \sa. Since all edges in $\mathcal{X}$ have at most two crossings, adopting the proof of Lemma~\ref{lem:2_planar_small_faces} we can prove that the endpoints of $(u,v)$, $(u',v')$ and $(u'',v'')$ define a \pec $\mathcal{C}$ on five vertices, with at most five edges passing through its interior. We proceed by redrawing these five edges as chords of $\mathcal{C}$ (as in Figure~\ref{fig:2_planar_5gon}). All its boundary edges are true-planar in the new drawing. The derived graph is optimal, as it has at least as many edges as $G$. Observe, however, that $\mathcal{C}$ becomes a true-planar $5$-cycle in the new drawing; a contradiction to Property~\ref{prp:3planar_even_order}.
\end{proof}

\begin{figure}[t]
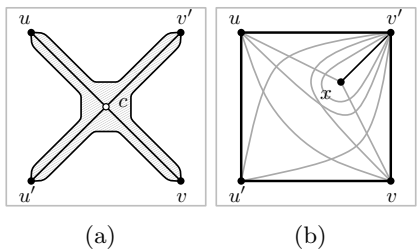

	\centering
    \begin{minipage}[b]{.18\textwidth}
        \centering
        \includegraphics[width=\textwidth,page=1]{images/3planar_one_crossing}
        \subcaption{~}\label{fig:3_planar_one_crossing_before}
    \end{minipage}
    \begin{minipage}[b]{.18\textwidth}
        \centering
        \includegraphics[width=\textwidth,page=2]{images/3planar_one_crossing}
        \subcaption{~}\label{fig:3_planar_one_crossing_after}
    \end{minipage}
    \caption{%
    Configurations used in Lemma~\ref{lem:3_planar_small_faces_2}.}.
    \label{fig:3_planar_one_crossing_1}
\end{figure}

\noindent The proof of Lemma~\ref{prp:3_planar_skeleton_connected} is similar to the one of Lemma~\ref{prp:2_planar_skeleton_connected} and is therefore omitted.

\begin{lemma}
The true planar skeleton $\Pi(G)$ of a PMCM-drawing $\Gamma(G)$ of an optimal $3$-planar graph is connected.
\label{prp:3_planar_skeleton_connected}
\end{lemma}

\begin{lemma}
The true-planar skeleton $\Pi(G)$ of a PMCM-drawing $\Gamma(G)$ of an optimal $3$-planar graph $G$ contains only faces of length $6$, each of which contains $8$ crossing edges in $\Gamma(G)$.
\label{lem:3_planar_faces}
\end{lemma}
\begin{proof}
Since by Lemma~\ref{prp:3_planar_skeleton_connected} $\Pi(G)$ is connected, all faces of $\Pi(G)$ are connected as well. By Lemma~\ref{lem:3_planar_small_faces}, any edge that is crossed three times in $\Gamma(G)$ is a chord of an empty true-planar $6$-cycle in $\Gamma(G)$. By Lemma~\ref{lem:3_planar_small_faces_2}, an edge $e$ that is crossed fewer than three times belongs to a crossing component $\mathcal{X}(e)$ containing an edge that is crossed three times. This last edge defines an empty true-planar $6$-cycle in $\Gamma(G)$ and by the observation we made in the proof of Lemma~\ref{lem:3_planar_xing_comp} all edges of $\mathcal{X}(e)$, including $e$, are also chords of this cycle. So, every crossing edge is a chord of a true-planar $6$-cycle. Note that one cannot embed nine edges in the interior of a true-planar $6$-cycle without deviating $3$-planarity but at most eight. We claim that $\Pi(G)$ has no chordless faces. First, we observe that $\Pi(G)$ cannot contain a chordless face of size $\geq 4$, as otherwise we could draw in its interior at least one chord, which would contradict the optimality of $G$. Also, by Property~\ref{prp:3planar_odd_cycle} $\Pi(G)$ contains no faces of length $3$. Finally, observe that $\Pi(G)$ cannot contain faces of length $1$ or $2$, as those would correspond to homotopic self-loops and parallel edges. This completes the proof.
\end{proof}

\noindent We say that a chord of a cycle of length $2s$ is a \emph{middle chord} if the two paths along the cycle connecting its endpoints both have length $s$. Next we state the main theorem of this section.

\begin{theorem}
A graph $G$ is optimal $3$-planar if and only if $G$ admits a drawing $\Gamma(G)$ without homotopic parallel edges and self-loops, such that the true-planar skeleton $\Pi(G)$ of $\Gamma(G)$ spans all vertices of $G$, it contains only faces of length $6$ (that are not necessarily simple), and each face of $\Pi(G)$ has $8$ crossing edges in its interior in $\Gamma(G)$ such that one of the middle chords is missing.
\label{thm:3-characterization}
\end{theorem}
\begin{proof}
For the forward direction, consider an optimal $3$-planar graph $G$. By Lemma~\ref{lem:3_planar_faces}, the true-planar skeleton $\Pi(G)$ of its $3$-planar PMCM-drawing $\Gamma(G)$ contains only faces of length $6$ and each face of $\Pi(G)$ has $8$ crossing edges in its interior in $\Gamma(G)$. By Property~\ref{prp:3_planar_quasi}, one of the three middle chords of each face of $\Pi(G)$ cannot be present. Since the endpoints of two crossing edges are within a true-planar $6$-cycle (by Lemmas~\ref{lem:3_planar_small_faces} and \ref{lem:3_planar_small_faces_2}) and since $\Pi(G)$ is connected (by Lemma~\ref{prp:3_planar_skeleton_connected}), $\Pi(G)$ spans all vertices of $G$. This completes the proof of this direction.

For the reverse direction, denote by $n$, $m$ and $f$ the number of vertices, edges and faces of $\Pi(G)$. Since $\Pi(G)$ spans all vertices of $G$, it suffices to prove that $G$ has exactly $5.5n-11$ edges. The fact that $\Pi(G)$ contains only faces of length $6$ implies that $6f=2m$. By Euler's formula for planar graphs, $m=3(n-2)/2$ and $f=(n-2)/2$ follows. Since each face of $\Pi(G)$ contains exactly $8$ crossing edges, the total number of edge of $G$ equals to $m+8f=5.5n-11$.
\end{proof}

% ======================================================================
\section{Further Insights From Our Work}
\label{sec:discussion}
% ======================================================================
In this section, we give new insights which follow from the new characterization of optimal $2$- and $3$-planar graphs. For simple optimal $3$-planar graphs we can note the following. Since the planar skeleton of an optimal $3$-planar graph consists exclusively of faces of length $6$, it cannot be simple. Hence, simple $3$-planar graphs do not reach the bound of $5.5n -11$ edges. Note that the best-known lower bound for simple optimal $3$-planar graph is $5.5n-15$~\cite{PachT97}.

\begin{crl}
Simple $3$-planar graphs have at most $5.5n - 11.5$ edges.
\label{rec:simple3-planar}
\end{crl}
\medskip

A \emph{bar-visibility representation} of a graph is a representation where vertices are represented as horizontal bars, and edges as vertical segments, called \emph{visibilities}, between corresponding bars. In the traditional bar-visibility model, a visibility edge is not allowed to cross any other bar except for the two bars at its endpoints. A central result here is due to Tamassia and Tollis~\cite{DBLP:journals/dcg/TamassiaT86} who showed that any biconnected planar graph admits a bar-visibility representation, which can be computed in linear time. The variant of \emph{bar 1-visibility} allows each visibility edge to cross at most one vertex bar. This model allows to represent also non-planar graphs in a limited way, e.g., the number of edges of a bar 1-visible graph on $n$ vertices can be at most $6n-20$~\cite{DBLP:journals/jgaa/DeanEGLST07}. Notable is a result by Brandenburg~\cite{DBLP:journals/jgaa/Brandenburg14} who showed that $1$-planar graphs admit bar 1-visibility representations; see also~\cite{DBLP:journals/jgaa/Evans0LMW14}.

We follow a similar technique to the one of Brandenburg~\cite{DBLP:journals/jgaa/Brandenburg14} to prove that simple optimal $2$-planar graphs are bar 1-visible. Since the faces defined by the true-planar skeleton $\Pi(G)$ of a simple optimal $2$-planar graph $G$ have size $5$, we can construct a bar-visibility representation $\mathcal{L}(G)$ of $\Pi(G)$ based on an $s$-$t$ ordering of $\Pi(G)$~\cite{DBLP:journals/dcg/TamassiaT86}. In the $s$-$t$ ordering each face is oriented such that it consists of a source and a target vertex joined by two chains of vertices (one on the left and one on the right). Since $\Pi(G)$ consists of faces of length $5$, the two chains have either $1$ and $2$ vertices each, or, $0$ and $3$ vertices each. In $\mathcal{L}(G)$, the source and target bars of a face $f$ see each other through a vertical visibility edge $b_f$ and the bars of the two chains are arranged to the left and to the right of $b_f$. Now it is straightforward to extend the bars of the two chains towards $b_f$, such that the bars of the two chains are vertically overlapping, and all five crossing edges of that face are realized. We conclude this observation in the following corollary.

\begin{crl}
Simple optimal $2$-planar graphs admit bar $1$-visibility representations.
\label{rec:2-planar-bar1}
\end{crl}
\medskip

In a \emph{fan-planar drawing} of a graph an edge can cross only edges with a common endpoint. Graphs that admit fan-planar drawings are called \emph{fan-planar}. Fan-planar graphs have been introduced by Kaufmann and Ueckerdt~\cite{KU14}, who proved that every simple $n$-vertex fan-planar drawing has at most $5n-10$ edges, and that this bound is tight for $n \geq 20$. This density result immediately implies that optimal $3$-planar graphs are not fan-planar. On the other hand, the density bound of $2$-planar graphs is the same as the one for fan-planar graphs. Binucci et al.~\cite{DBLP:journals/tcs/BinucciGDMPST15} already investigated the relationship between these two classes and in particular they proved that there exist $2$-planar graphs that are not fan-planar. Our characterization for simple optimal $2$-planar graphs, however, implies that all optimal $2$-planar graphs are fan-planar, as their PMCM-drawings are in fact fan-planar. We conclude this observation in the following corollary.

\begin{crl}
Simple optimal $2$-planar graphs are optimal fan-planar.
\label{rec:2-planar-fan}
\end{crl}
\medskip

\noindent Our characterizations naturally lead to many open questions. In the following we name a few.
\begin{itemize}
\item What is the complexity of the recognition problem for optimal $2$- and $3$-planar graphs?
\item What is the exact upper bound on the number of edges of simple optimal $3$-planar graphs? We conjecture that they do not have more than $5.5n-15$ edges.
\item Theorems~\ref{thm:2-characterization} and~\ref{thm:3-characterization} imply that optimal $2$- and $3$-planar graphs have a fully triangulated planar subgraph. Can this property be proved for optimal $4$-planar or more in general for optimal $k$-planar graphs? Proving this property would be useful to derive better density bounds for $k\geq4$.
\item By Properties~\ref{prp:2_planar_quasi} and \ref{prp:3_planar_quasi}, optimal $2$- and $3$-planar graphs are quasi-planar. Angelini et al.~\cite{ABBL17} proved that every simple $k$-planar graph is $(k+1)$-quasi planar for $k \geq 3$ (i.e., it can be  drawn with no $k+1$ pairwise crossing edges). Our results about optimal $2$-planar and even more about optimal $3$-planar graphs give indications that the result by Angelini et al.~\cite{ABBL17} may hold also for $k=2$.
\item We have found a RAC drawing (i.e., a drawing in which all crossing edges form right angles) with at most one bend per edge for the optimal $2$-planar graph obtained from the dodecahedron as its true-planar structure. Is this generalizable to all simple optimal $2$-planar graphs?
\end{itemize}

\bibliographystyle{abbrv}
\bibliography{references}

\end{document}